\documentclass{svmult}
\smartqed

\usepackage{amsmath,amssymb,natbib,color,bbm,ifthen,graphicx}
\usepackage[latin1]{inputenc}
\usepackage{bbm}
\usepackage{pstricks}

\newenvironment{enum_W}
  {%
  \setlength{\leftmargini}{4em}\begin{enumerate}}
  {\end{enumerate}}

  {\end{enumerate}}

  {\end{enumerate}}

\def\rme{\mathrm{e}}

\def\dj{u}
\def\dk{v}
\def\rme{\mathrm{e}}
\def\rmi{\mathrm{i}}
\def\ltwo{\mathrm{L}^2}

\def\cp{\stackrel{P}{\longrightarrow}}
\def\cl{\stackrel{\mathcal{L}}{\longrightarrow}}

\def\dwt{W}
\def\bdwt{\mathbf{\dwt}}
\def\indexset{\mathcal{I}}

\def\L{\mathrm{T}}

\def\be{\mathbf{e}}

\def\diffop{\mathbf{\Delta}}

\def\1{\mathbbm{1}}

\def\bx{\mathbf{x}}

\def\Zset{\mathbb{Z}}

\def\Rset{\mathbb{R}} 
\def\PE{\mathbb{E}} 
\def\PVar{\mathrm{Var}}
\def\PCov{\mathrm{Cov}}

\def\Cum{\mathrm{Cum}}

\def\ie{\textit{i.e.} }

\def\argmin{\mathop{\mathrm{Argmin}}}
\def\prob{\mathbb{P}}

\newcommand{\eqdef}{\ensuremath{\stackrel{\mathrm{def}}{=}}}
\newcommand{\eqsp}{\;}

\newcommand{\AVvar}[3][]
{
\ifthenelse{\equal{#1}{}}{\mathbf{V}_{#3}(#2)}{\mathbf{V}_{#3}(#2,#1)}}
\newcommand{\AVvarJoint}[3][]
{
\ifthenelse{\equal{#1}{}}{\mathbf{W}_{#3}(#2)}{\mathbf{W}_{#3}(#2,#1)}}

\newcommand{\AsympVarWWE}[2][]
{\mathrm{V}(#2)}

\newcommand{\AVvarInv}[3][]
{\ifthenelse{\equal{#1}{}}{\mathbf{V}^{-1}_{#3}(#2)}{\mathbf{V}^{-1}_{#3}(#2,#1)}}
\newcommand{\sigmaasymp}[2][]
{
\ifthenelse{\equal{#1}{}}{\sigma(#2)}{\sigma(#2,#1)}}

\def\vjsymb{\sigma}
\newcommand{\vj}[4][]{%
\ifthenelse{\equal{#1}{}}{\vjsymb^{#4}_{#2}(#3)}{\vjsymb^{#4}_{#2}(#3,#1)}}
\newcommand{\stdj}[3][]{%
\ifthenelse{\equal{#1}{}}{\vjsymb_{#2}(#3)}{\vjsymb_{#2}(#3,#1)}}

\def\rmd{\mathrm{d}}

\newcommand{\hvj}[3][]{%
\ifthenelse{\equal{#1}{}}{\hat{\vjsymb}^2_{#2}}{\hat{\vjsymb}^2_{#2}(#1)}}

\def\densletter{\mathbf{D}}
\newcommand{\bdens}[4][]{%
\densletter_{#2}({#3};#4)}

\def\cum{\mathcal{K}}

\def\allWA{\ref{item:Wreg}--\ref{item:Wvstd}}

\renewcommand{\hat}{\widehat}
\renewcommand{\tilde}{\widetilde}

\newcommand\pent[1]{\left\lfloor #1 \right\rfloor}

\newcommand\den[1]{\left|\mathbf{D}_{#1,0}(\lambda;f)\right|^2d\lambda}

\def\cov{\mathrm{Cov}}
\def\Cum{\mathrm{Cum}}
\def\Sbar{\bar{\sigma}^2_{j}}
\newcommand\somme[2]{\sum\limits_{#1}^{#2}}

\def \2{2^{J_2-j}}
\def \2p{2^{J_2-j'}}

\newcommand{\scalogram}[1]{\widehat{\sigma}_{#1}^2}

\date{\today}
\title{Testing for homogeneity of variance in the wavelet domain.}
\author{O.~Kouamo\inst{1} \and E.~Moulines\inst{2} \and F.~Roueff \inst{2} }

\institute{ENSP, LIMSS, BP : 8390 Yaoundé, \texttt{olaf.kouamo@telecom-paristech.fr} \and Institut Télécom, Télécom ParisTech, CNRS UMR 5181, Paris, \texttt{eric.moulines,francois.roueff@telecom-paristech.fr}}

\begin{document}

\maketitle
\begin{abstract}
The danger of confusing long-range dependence with non-stationarity
has been pointed out by many authors. Finding an answer to this difficult question is  of importance to model time-series showing trend-like behavior, such as  river run-off in hydrology, historical temperatures in the study of climates changes, or packet counts in network traffic engineering.

The main goal of this paper is to develop a test procedure to detect the presence of non-stationarity for a class of processes whose $K$-th order difference is stationary. Contrary to most of the proposed methods, the test procedure has the same distribution for short-range and long-range dependence covariance stationary processes, which means that this test is able to detect the presence of non-stationarity for processes showing long-range dependence or which are unit root.

The proposed test is formulated in the wavelet domain, where a change in the generalized spectral density results in a change in the variance of wavelet coefficients at one or several scales. Such tests have been already proposed in
 \cite{whitcher:2001}, but these authors do not have taken into account the dependence of the wavelet coefficients within scales and between scales. Therefore, the asymptotic distribution of the test they have proposed was erroneous; as a consequence,  the level of the test under the null hypothesis of stationarity was wrong.

 In this contribution, we introduce two test procedures, both using  an estimator of the variance of the scalogram at one or several scales. The asymptotic distribution of the test under the null is rigorously justified. The pointwise consistency of the test in the presence of a single jump in the general spectral density is also be presented.

A limited Monte-Carlo experiment is performed to illustrate our findings.
\end{abstract}
\section{Introduction}

For time series of short duration, stationarity and short-range dependence have usually been regarded to be approximately valid.
However, such an assumption becomes questionable  in the large data sets currently investigated in geophysics, hydrology or financial econometrics.
There has been a long lasting controversy to decide whether the deviations to ``short memory stationarity'' should be attributed to long-range dependence or are related to the presence of breakpoints in the mean, the variance, the covariance function or other types of more sophisticated structural changes.
The links between non-stationarity and long-range dependence (LRD) have been pointed out by many authors in the hydrology literature long ago:
\cite{klemes:1974} and \cite{boes:salas:1978} show that non-stationarity in the mean provides a possible
explanations of the so-called Hurst phenomenon. \cite{potter:1976} and later \cite{rao:yu:1986} suggested that more sophisticated changes may occur, and have proposed a method to detect such changes. The possible confusions between long-memory and some forms of nonstationarity have been discussed in the applied probability literature:
\cite{bhattacharya:gupta:waymire:1983} show that long-range dependence may be confused with the presence of a small monotonic trend.
This phenomenon has also been discussed in the econometrics literature. \cite{hidalgo:robinson:1996} proposed a test of presence of structural change in a long memory environment. \cite{granger:hyung:1999}  showed that linear processes with breaks can  mimic the autocovariance structure of a linear fractionally integrated long-memory process (a stationary process that encounters occasional regime
switches will have some properties that are similar to those of a long-memory process). Similar behaviors are considered in
\cite{diebold:inoue:2001} who provided simple and intuitive econometric models showing  that long-memory and structural
changes are easily confused. \cite{mikosch:starica:2004} asserted
that what had been seen by many authors as long memory in the volatility of the absolute values or the square of the log-returns
might, in fact, be explained by abrupt changes in the parameters of an underlying
GARCH-type models. \cite{berkes:horvatz:2006} proposed a testing procedure for distinguishing between a weakly dependent time
series with change-points in the mean and a long-range dependent time series. \cite{hurvich:lang:soulier:2005} have proposed
a test procedure for detecting long memory in presence of deterministic trends.

The procedure described in this paper deals with the problem of detecting changes which may occur in the
spectral content of a process. We will consider a process $X$ which, before and after the change, is not necessary stationary
but whose difference of at least a given order is stationary, so that polynomial trends up to that order can be discarded.
Denote by $\diffop X$ the first order difference of $X$,
$$
[\diffop X]_n \eqdef X_n - X_{n-1},\quad n\in\Zset\;,
$$
and define, for an integer $K \geq 1$ , the $K$-th order difference recursively as follows: $\diffop^K = \diffop \circ \diffop^{K-1}$.
A process $X$ is said to be $K$-th order difference stationary if $\diffop^K X$ is covariance stationary.
Let $f$ be a non-negative $2\pi$-periodic symmetric function such that there exists an integer $K$ satisfying,
$\int_{-\pi}^\pi |1-\rme^{-\rmi\lambda}|^{2K} f(\lambda) \rmd \lambda < \infty$.
We say that the process $X$ admits \emph{generalized  spectral density} $f$ if $\diffop^K X$
is weakly stationary and with spectral density function
\begin{equation}
\label{eq:generalized-spectral-density}
f_K(\lambda)=|1-\rme^{-\rmi\lambda}|^{2K}f(\lambda) \eqsp.
\end{equation}
This class of process include both short-range dependent
and long-range dependent processes, but also unit-root and fractional unit-root processes.
The main goal of this paper is to develop a testing procedure for distinguishing between a $K$-th order stationary process
and a non-stationary process.

In this paper, we consider the so-called \emph{a posteriori} or \emph{retrospective} method (see
\cite[Chapter 3]{brodsky:darkhovsky:2000}). The proposed test is formulated in the wavelet domain, where a change in the
generalized spectral density results in a change in the variance of the wavelet
coefficients. Our test is based on a CUSUM statistic, which is perhaps the most extensively used statistic for detecting and
estimating change-points in mean. In our procedure, the CUSUM is applied to the partial sums of the squared wavelet coefficients at a given scale or on a specific
range of scales. This procedure extends the test introduced in \cite{inclan:tiao:1994} to detect changes in the variance of
an independent sequence of random variables. To describe the idea, suppose that, under the null hypothesis, the time series
is $K$-th order difference stationary and that, under the alternative, there is one breakpoint where the generalized spectral
density of the process changes. We consider the scalogram in the range of scale $J_1,J_1+1,\dots,J_2$. Under
the null hypothesis, there is no change
in the variance of the wavelet coefficients at any given scale $j \in\{J_1,\dots,J_2\}$. Under the alternative,
these variances takes different values before and after the change point. The amplitude of the change depends on the scale,
and the change of the generalized spectral density.
We consider the $(J_2-J_1+1)$-dimensional  W2-CUSUM statistic $\{T_{J_1,J_2}(t),\,t \in [0,1]\}$
defined by \eqref{eq:cusum-multiple}, which is a CUSUM-like statistics applied to the square of the wavelet coefficients.
Using $T_{J_1,J_2}(t)$ we can construct an estimator $\hat{\tau}_{J_1,J_2}$ of the change point (no matter if a change-point
exists or not), by minimizing an appropriate norm of the W2-CUSUM statistics, $\hat{\tau}_{J_1,J_2}= \argmin_{t \in [0,1]}
\|T_{J_1,J_2}(t)\|_\star$.  The
statistic $T_{J_1,J_2}(\hat{\tau}_{J_1,J_2})$ converges to a well-know distribution under the null hypothesis (see Theorems
\ref{theo:main-result-singlescale} and \ref{theo:main-results-multiplescale}) but diverges to infinity under the
alternative (Theorems~\ref{theo:power-single-case} and~\ref{theo:power-multi-case}). A similar idea has been proposed by
\cite{whitcher:2001} but these authors did not take into account
the dependence of wavelet coefficient, resulting in an erroneous normalization and asymptotic distributions.

The paper is organized as follows.
In Section \ref{sec:wavelet-setting}, we introduce the wavelet setting and the relationship between the generalized spectral
density and the variance of wavelet coefficients at a given scale.
In Section \ref{sec:AsymptoticDistributionW2CUSUM}, our main assumptions are formulated and
the asymptotic distribution of the W2-CUSUM statistics  is  presented first
in the single scale (sub-section \ref{sec:single-scale}) and then in the multiple scales (sub-section \ref{sec:multiple-scale}) cases. In Section \ref{sec:test-statistic}, several possible test procedures
are described to detect the presence of changes at a single scale or simultaneously  at several
scales. In Section \ref{sec:applications}, finite sample performance of the test procedure is studied based on Monte-Carlo experiments.

\section{The wavelet transform of $K$-th order difference stationary processes}
\label{sec:wavelet-setting}
In this section, we introduce the wavelet setting, define the scalogram and explain how spectral change-points
can be observed in the wavelet domain. The main advantage of using the wavelet domain is to alleviate problems arising
when the time series exhibit is long range dependent.  We will recall some basic results obtained in \cite{moulines:roueff:taqqu:2007:jtsa} to support our claims. We refer the reader to that paper for the proofs of the stated results.

\noindent\textbf{The wavelet setting.}
The wavelet setting involves two functions $\phi$ and $\psi$ and their Fourier transforms
\[
\hat{\phi}(\xi) \eqdef \int_{-\infty}^\infty \phi(t) \rme^{- \rmi \xi t}\,dt \quad
\text{and}
\quad
\hat{\psi}(\xi) \eqdef \int_{-\infty}^\infty \psi(t) \rme^{- \rmi \xi t}\,dt,
\]
and assume the following:
\begin{enum_W}
\item\label{item:Wreg} $\phi$ and $\psi$ are compactly-supported, integrable, and $\hat{\phi}(0) = \int_{-\infty}^\infty \phi(t)\,dt = 1$ and $\int_{-\infty}^\infty \psi^2(t)\,dt = 1$.
\item\label{item:psiHat}
There exists $\alpha>1$ such that
$\sup_{\xi\in\Rset}|\hat{\psi}(\xi)|\,(1+|\xi|)^{\alpha} <\infty$.
\item\label{item:MVM} The function $\psi$ has  $M$ vanishing moments, \ie\ $ \int_{-\infty}^\infty t^m \psi(t) \,dt=0$ for all $m=0,\dots,M-1$
\item\label{item:MIM} The function $ \sum_{k\in\Zset} k^m\phi(\cdot-k)$
is a polynomial of degree $m$ for all $m=0,\dots,M-1$.
\end{enum_W}\def\allWA{\ref{item:Wreg}-\ref{item:MIM}}
The fact that both $\phi$ and $\psi$ have finite support (Condition~\ref{item:Wreg}) ensures that the corresponding filters
(see~(\ref{eq:FilterJ})) have finite impulse responses (see~(\ref{eq:FilterJ-FIR})). While the support of the Fourier
transform of $\psi$ is the whole real line, Condition~\ref{item:psiHat} ensures that this Fourier transform decreases quickly
to zero. Condition~\ref{item:MVM} is an important characteristic of wavelets: it ensures that they oscillate and that their
scalar product with continuous-time polynomials up to degree $M-1$ vanishes.
Daubechies wavelets and Coiflets having at least two vanishing moments satisfy these conditions.


Viewing the wavelet $\psi(t)$ as a basic template, define the family $\{\psi_{j,k}, j \in \Zset, k \in \Zset\}$ of translated and dilated functions
\begin{equation}\label{eq:psiJK}
\psi_{j,k}(t)=2^{-j/2}\,\psi(2^{-j}t-k) ,\quad j\in\Zset,\, k\in\Zset \eqsp .
\end{equation}
Positive values of $k$ translate $\psi$ to the right, negative values to the left. The \emph{scale index} $j$ dilates $\psi$
so that large values of $j$ correspond to coarse scales and hence to low frequencies.

Assumptions \allWA\ are standard in the context of a multiresolution analysis (MRA) in which case, $\phi$ is the
scaling function and $\psi$ is the associated wavelet, see for instance \cite{mallat:1998,cohen:2003}.
Daubechies wavelets and Coiflets are examples of orthogonal wavelets constructed using an MRA.
In this paper, we do not assume the wavelets to be orthonormal nor
that they are associated to a multiresolution analysis. We may therefore work with other convenient choices for $\phi$ and $\psi$
as long as~\allWA\ are satisfied.

\noindent\textbf{Discrete Wavelet Transform (DWT) in discrete time.}
We now describe how the wavelet coefficients are defined in discrete time, that is for a real-valued sequence
$\{x_k,\,k\in\Zset\}$ and for a finite sample $\{x_k,\,k=1,\dots,n\}$. Using the scaling function $\phi$, we first
interpolate these discrete values to construct the following continuous-time functions
\begin{equation}\label{eq:bX}
\bx_n(t) \eqdef \sum_{k=1}^n x_k \,\phi(t-k) \quad\text{and}\quad \bx(t) \eqdef \sum_{k\in\Zset} x_k\, \phi(t-k), \quad
t\in\Rset \; .
\end{equation}
Without loss of generality we may suppose that the support of the scaling function $\phi$ is included in $[-\L,0]$ for some
integer $\L\geq1$. Then
$$
\bx_n(t)=\bx(t)\quad\text{for all}\quad t\in[0, n-\L+1]\;.
$$
We may also suppose that the support of the wavelet function $\psi$ is included in $[0,\L]$. With these conventions, the support of $\psi_{j,k}$ is included in the
interval $[2^j k, 2^j(k+\L)]$.
Let $\tau_0$ be an arbitrary shift order.
The wavelet coefficient $\dwt^\bx_{j,k}$ at scale $j\geq0$ and location $k\in\Zset$ is formally defined
as the scalar product in $\ltwo(\mathbb{R})$ of the function $t \mapsto \bx(t)$ and the wavelet $t \mapsto \psi_{j,k}(t)$:
\begin{equation}\label{eq:coeffN}
\dwt^\bx_{j,k} \eqdef \int_{-\infty}^\infty \bx(t) \psi_{j,k}(t)\,dt
=\int_{-\infty}^\infty \bx_n(t) \psi_{j,k}(t)\,dt,
 \quad  j \geq 0, k \in \Zset \;,
\end{equation}
when $[2^j k, 2^j k+\L)]\subseteq [0, n-\L+1]$, that is, for all $(j,k)\in\indexset_n$, where
$\indexset_n= \{(j,k):\,j\geq0, 0 \leq k  < n_j \}$ with $n_j= 2^{-j}(n-\L+1)-\L+1$.
It is important to observe that the definition of the wavelet coefficient $\dwt_{j,k}$ at a given index $(j,k)$ does not depend
on the sample size $n$ (this is in sharp contrast with Fourier coefficients).
For ease of presentation, we will use the convention that at each scale $j$, the first available  wavelet coefficient
$\dwt_{j,k}$ is indexed by $k=0$, that is,
\begin{equation}
\label{eq:deltan}
\indexset_n \eqdef \{(j,k):\,j\geq0, 1\leq k \leq n_j \}\quad\text{with}\quad n_j= 2^{-j}(n-\L+1)-\L+1 \eqsp.
\end{equation}

\noindent\textbf{Practical implementation.}
In practice the DWT of  $\{x_k,\,k=1,\dots,n\}$ is not computed using~(\ref{eq:coeffN}) but by linear filtering and
decimation. Indeed the wavelet coefficient $\dwt^\bx_{j,k}$ can be expressed  as
\begin{equation}\label{eq:down}
\dwt^{\bx}_{j,k}=\sum_{l\in\Zset} x_l\,h_{j,2^jk-l},\quad (j,k)\in\indexset_n ;,
\end{equation}
where
\begin{equation}\label{eq:FilterJ}
 h_{j,l} \eqdef 2^{-j/2} \int_{-\infty}^\infty \phi(t+l)\psi(2^{-j}t)\,dt \; .
\end{equation}
For all $j\geq0$, the discrete Fourier transform of the transfer function $\{ h_{j,l} \}_{l \in \Zset}$  is
\begin{equation}\label{eq:fjdef}
H_j(\lambda) \eqdef \sum_{l\in\Zset}h_{j,l}\rme^{-\rmi\lambda l} =2^{-j/2} \int_{-\infty}^\infty
\sum_{l\in\Zset}\phi(t+l)\rme^{-\rmi\lambda l} \psi(2^{-j}t)\,dt.
\end{equation}
Since $\phi$ and $\psi$ have compact support, the sum in~(\ref{eq:fjdef}) has only a finite number of non-vanishing terms and,
$H_j(\lambda)$ is the transfer function of a finite impulse response filter,
\begin{equation}\label{eq:FilterJ-FIR}
H_j(\lambda)=\sum_{l=-\L(2^j+1)+1}^{-1} h_{j,l}\rme^{-\rmi\lambda l} \; .
\end{equation}
When $\phi$ and $\psi$ are the scaling and the wavelet functions associated to a MRA, the wavelet coefficients may be obtained recursively
by applying a finite order filter and downsampling by an order $2$. This recursive procedure
is referred to as \emph{the pyramidal   algorithm}, see for instance \cite{mallat:1998}.

\noindent\textbf{The  wavelet spectrum and the scalogram.}
Let $X=\{X_t,\;t\in\Zset\}$ be a real-valued process with  wavelet coefficients $\{\dwt_{j,k},\,k\in\Zset\}$ and define
$$
\sigma_{j,k}^2=\PVar(\dwt_{j,k}) \; .
$$
If $\diffop^M X$ is stationary, by Eq~$(16)$
in~\cite{moulines:roueff:taqqu:2007:jtsa}, we have that, for all $j$, the process of its wavelet coefficients at scale $j$,
$\{\dwt_{j,k},\,k\in\Zset\}$,  is also stationary. Then, the wavelet variance $\sigma_{j,k}^2$ does not depend on $k$, $\sigma^2_{j,k}= \sigma_j^2$. The sequence $(\sigma_j^2)_{j\geq0}$ is called the \emph{wavelet spectrum} of the process $X$.

If moreover  $\diffop^M X$ is centered, the wavelet spectrum can be estimated by using the scalogram, defined as the empirical mean
of the squared wavelet coefficients computed from the sample $X_1,\dots, X_n$:
$$
\scalogram{j}=\frac1{n_j}\sum_{k=1}^{n_j} \dwt_{j,k}^2 \; .
$$
By~\cite[Proposition~1]{moulines:roueff:taqqu:2007:jtsa}, if $K\leq M$, then the scalogram of $X$ can be expressed using
the generalized spectral density $f$ appearing in~(\ref{eq:generalized-spectral-density}) and the filters $H_j$ defining the
DWT in~(\ref{eq:fjdef}) as follows:
\begin{equation}\label{eq:scalogramVSgenspecdens}
\sigma_j^2 =\int_{-\pi}^\pi \left|H_j(\lambda)\right|^2\,f(\lambda)\,\rmd\lambda ,\quad j\geq 0\;.
\end{equation}

\section{Asymptotic distribution of the W2-CUSUM statistics}
\label{sec:AsymptoticDistributionW2CUSUM}
\subsection{The single-scale case}
\label{sec:single-scale}

To start with simple presentation and statement of results, we first focus in this section on a test procedure aimed at detecting
a change in the variance of the wavelet coefficients  at a single scale  $j$.
Let $X_1,\ldots,X_n$ be the $n$ observations of a time series, and denote by  $W_{j,k}$  for
$(j,k)\in\indexset_n$  with $\indexset_n $ defined in \eqref{eq:deltan} the associated  wavelet coefficients.
In view of~\eqref{eq:scalogramVSgenspecdens}, if $X_1,\dots,X_n$ are a $n$ successive observations of a $K$-th order
difference stationary process, then the wavelet variance at each given scale $j$ should be constant. If the process $X$ is not
$K$-th order stationary, then it can be expected that the wavelet variance will change either gradually or abruptly (if there is
a shock in the original time-series). This thus suggests to investigate the consistency of the variance of the wavelet coefficients.

There are many works aimed at detecting the change point in the variance of a sequence of independent random variables;
such problem has also been considered, but much less frequently,
for sequences of dependent variables. Here, under the null assumption of $K$-th order difference stationarity, the wavelet coefficients $\{\dwt_{j,k},\;k\in\Zset\}$ is a covariance  stationary sequence whose spectral density is given by (see~\cite[Corollary~1]{moulines:roueff:taqqu:2007:jtsa})
\begin{equation}
\label{density1scale}
\bdens[\phi,\psi]{j,0}{\lambda}{f} \eqdef
\sum_{l=0}^{2^{j}-1}   f(2^{-j}(\lambda+2l\pi))\,
2^{-j} \left|H_{j}(2^{-j}(\lambda+2l\pi))\right|^2 \eqsp.
\end{equation}
We will adapt the approach developed in \cite{inclan:tiao:1994}, which uses cumulative sum (CUSUM) of squares to detect change points in the variance.

In order to define the test statistic, we first introduce a change point estimator for the mean of the square of the wavelet coefficients
at each scale $j$.
\begin{equation}\label{est-change-time}
\hat{k}_j=\underset{1\leq k\leq n_j}{\mathrm{argmax}}\left|\sum\limits_{1\leq i\leq k}W_{j,i}^2-\frac{k}{n_j}\sum\limits_{1\leq i\leq n_j}W_{j,i}^2\right| \eqsp.
\end{equation}
Using this change point estimator, the  W2-CUSUM statistics is defined as
\begin{equation}
\label{eq:cusum}
T_{n_j}=\frac{1}{n_j^{1/2}s_{j,n_j}}\left|\sum_{1\leq i\leq \hat{k}_j}W_{j,i}^2-\frac{\hat{k}_j}{n_j}\sum_{1\leq i\leq n_j}W_{j,i}^2\right| \eqsp,
\end{equation}
where $s_{j,n_j}^2$ is a suitable estimator of the variance of the sample mean of the $W_{j,i}^2$.
Because wavelet coefficients at a given scale are correlated, we use the Bartlett estimator of the variance, which is defined by
\begin{equation}\label{bartlett estimator}
s_{j,n_j}^2=\hat{\gamma}_{j}(0)+2\sum\limits_{1\leq l\leq q(n_j)}w_l\big(q(n_j)\big)\hat{\gamma}_{j}(l) \eqsp,
\end{equation}
where
\begin{equation}\label{eq:gamma}
\hat{\gamma}_{j}(l) \eqdef \frac{1}{n_j}\sum\limits_{1\leq i\leq n_j-l}(W_{j,i}^2-\scalogram{j})(W_{j,i+l}^2-\scalogram{j}),
\end{equation}
are the sample autocovariance of $\{W_{j,i}^2,\,i=1,\dots,n_j\}$, $\scalogram{j}$ is the scalogram and,
for a given integer $q$,
\begin{equation}\label{eq-bandwidth}
w_l(q)=1-\frac{l}{1+q} \eqsp, l \in \{0, \dots, q\} \eqsp
\end{equation}
are the so-called Bartlett weights.

The test differs from statistics proposed in \cite{inclan:tiao:1994} only in its denominator, which is the square root of a
consistent estimator of the partial sum's variance. If $\{ X_n \}$ is  short-range
dependent, the variance of the partial sum of the scalograms is not simply the sum of the
variances of the individual square wavelet coefficient, but also includes the autocovariances of these termes.
Therefore, the estimator of the averaged scalogram variance involves not only sums of squared deviations of the scalogram coefficients, but
also its weighted autocovariances up to lag $q(n_j)$. The weights $\{w_l(q(n_j))\}$ are those
suggested by \cite{newey:west:1987} and always yield a positive sequence of autocovariance, and
a positive estimator of the (unnormalized) wavelet spectrum at scale $j$, at frequency zero using a Bartlett window.
We will first established the consistency of the estimator $s_{j,n_j}^2$ of the variance of the scalogram at scale $j$ and the convergence
of the empirical process of the square wavelet coefficients to the Brownian motion.
Denote by $D([0,1])$ is the Skorokhod space of functions which are right continuous at each point of $[0,1)$ with left limit of $(0,1]$ (or
\emph{cadlag} functions). This space is, in the sequel, equipped with the classical Skorokhod metric.
%
\begin{theorem} \label{theo:main-result-singlescale}
Suppose that  $X$ is a Gaussian process with generalized spectral density $f$. Let $(\phi,\psi)$ be a scaling and a wavelet function satisfying \ref{item:Wreg}-\ref{item:MIM}. Let $\{ q(n_j) \}$ be a non decreasing sequence of integers satisfying
\begin{equation}
\label{eq:condition-q2}
q(n_j)\rightarrow\infty \;\;\;\text{and} \;\;\;q(n_j)/n_j \rightarrow 0 \;\; as \;\;n_j\rightarrow \infty.
\end{equation}
Assume that $\Delta^MX$ is non-deterministic and centered, and that $\lambda^{2M}f(\lambda)$ is two times differentiable in $\lambda$ with bounded second order derivative. Then for any fixed scale $j$, as $n\rightarrow\infty,$
\begin{equation}\label{consistency}
s_{j,n_j}^2 \cp\frac{1}{\pi}\int_{-\pi}^{\pi}\den{j},
\end{equation}
where $\bdens[\phi,\psi]{j,0}{\lambda}{f}$ is the wavelet coefficients spectral density at scale $j$ see (\ref{density1scale}).
Moreover, defining $\sigma_j^2$ by~\eqref{eq:scalogramVSgenspecdens},
\begin{equation}\label{convergence}
\frac{1}{n_j^{1/2}s_{j,n_j}} \sum\limits_{i=1}^{[n_j t]} \left(W_{j,i}^2-\sigma_j^2\right) \cl B(t)\;\; \text{in}\;\;
D([0,1]), \quad \text{as} \quad n\rightarrow\infty
\end{equation}
where $(B(t), t \in [0,1])$ is the standard Brownian motion.
\end{theorem}
\begin{remark}
The fact that $X$ is Gaussian can be replaced by the more general assumption that the process $X$ is linear in the strong sense, under
appropriate moment conditions on the innovation. The proofs are then more involved, especially to establish the invariance principle which
is pivotal in our derivation.
\end{remark}
\begin{remark}
By allowing $q(n_j)$ to increase but at a slower rate than the number of
observations, the estimator of the averaged scalogram variance adjusts  appropriately for
general forms of short-range dependence among the scalogram coefficients. Of course, although the condition
\eqref{eq:condition-q2} ensure the consistency of $s_{j,n_j}^2$,
they provide little guidance in selecting a truncation lag $q(n_j)$.
When $q(n_j)$ becomes large relative to the sample size $n_j$, the finite-sample distribution of the test statistic
might be far from its asymptotic limit.
However $q(n_j)$ cannot be chosen too small since the
autocovariances beyond lag $q(n_j)$ may be significant  and should be included in the
weighted sum. Therefore, the truncation lag must be chosen ideally using some data-driven procedures.
\cite{andrews:1991}  and \cite{newey:west:1994} provide a data-dependent rule for choosing $q(n_j)$.
 These contributions suggest that selection of bandwidth according to an asymptotically optimal procedure tends to lead to more accurately sized test statistics than do traditional procedure
The methods suggested by \cite{andrews:1991} for selecting the bandwidth optimally is a plug-in approach.
This procedure  require the researcher to fit an  ARMA model of given order to provide a rough estimator of the spectral density and of its derivatives at zero frequencies (although misspecification of the order affects only optimality but not consistency). The minimax optimality of this method is  based on an asymptotic mean-squared error criterion and its behavior in the finite sample case is not precisely known.
The procedure outlined in \cite{newey:west:1994} suggests to bypass the modeling step, by using instead a pilot truncated kernel estimates of the spectral density and its derivative. We use these data driven procedures in the Monte Carlo experiments (these procedures have been implemented in the \emph{R-package sandwich}.
\end{remark}

\begin{proof}
Since $X$ is Gaussian and $\diffop^MX$ is centered, Eq.~(17) in \cite{moulines:roueff:taqqu:2007:jtsa} implies that $\{W_{j,k},\;k\in\Zset\}$ is a centered Gaussian process, whose distribution is determined by
$$
\gamma_j(h)=\cov(W_{j,0},W_{j,h})=\int_{-\pi}^{\pi}\bdens[\phi,\psi]{j,0}{\lambda}{f}e^{-i\lambda h}d\lambda\;.
$$
 From Corollary 1 and equation $(16)$ in \cite{moulines:roueff:taqqu:2007:jtsa}, we have
\begin{multline*}
\bdens[\phi,\psi]{j,0}{\lambda}{f}\\ =\sum\limits_{l=0}^{2^{j}-1}f\left(2^{-j}(\lambda+2l\pi)\right)2^{-j}\left|\tilde{H}_j(2^{-j}(\lambda+2l\pi))\right|^2\left|1-e^{-i2^{-j}(\lambda+2l\pi)}\right|^{2M},
\end{multline*}
where $\tilde{H}_j$ is a trigonometric polynomial.
Using that $$|1-e^{-i\xi}|^{2M}=|\xi|^{2M}\left|\frac{1-e^{-i\xi}}{i\xi}\right|^{2M}$$ and that $|\xi|^{2M}f(\xi)$ has a bounded second order derivative, we get that $\bdens[\phi,\psi]{j,0}{\lambda}{f}$ has also a bounded second order derivative. In particular, \begin{equation}\label{eq:DjProperties}
\int_{-\pi}^{\pi}\den{j}<\infty\quad\text{and}\quad
\sum_{s\in\Zset}|\gamma_j(s)|<\infty\;.
\end{equation}

The proof may be decomposed into 3 steps. We first prove the consistency of the Bartlett estimator of the variance of the squares of wavelet coefficients $ s^2_{j,n_j}$, that is (\ref{consistency}). Then we determine the asymptotic normality of the finite-dimensional distributions of the empirical scalogram, suitably centered and normalized. Finally a tightness criterion is proved, to establish the convergence in the
Skorokhod space. Combining these three steps completes the proof of \eqref{convergence}.

\noindent\textbf{Step~1}. Observe that, by the Gaussian property, $\cov(W^2_{j,0},W^2_{j,h})=2\gamma^2_j(h)$. Using Theorem 3-i in \cite{giraitis:kokoszka:2003}, the limit \eqref{consistency} follows from
\begin{equation}\label{autocov}
2\sum\limits_{h=-\infty}^{+\infty}\gamma^2_j(h)=\frac{1}{\pi}\int_{-\pi}^{\pi}\den{j}<\infty,
\end{equation}
and
\begin{equation}\label{cumulant}
\underset{h\in\Zset}{\sup}\sum\limits_{r,s=-\infty}^{+\infty}|\cum (h,r,s)|<\infty.
\end{equation}
where
\begin{equation}\label{eq:cum4}
\cum(h,r,s)=\Cum\big(W_{j,k}^2,W_{j,k+h}^2,W_{j,k+r}^2,W_{j,k+s}^2\big).
\end{equation}
Equation~(\ref{autocov}) follows from Parseval's equality and (\ref{eq:DjProperties}). Let us now prove (\ref{cumulant}). Using that the wavelet coefficients are Gaussian, we obtain
\begin{multline*}
\mathcal{K}(h,r,s)=12\big\{\gamma_j(h)\gamma_j(r-s)\gamma_j(h-r)\gamma_j(s)\\
+\gamma_j(h)\gamma_j(r-s)\gamma_j(h-s)\gamma_j(r)+\gamma_j(s-h)\gamma_j(r-h)\gamma_j(r)\gamma_j(s)\big\} \; .
\end{multline*}
The bound of the last term is given by
\begin{align*}
\underset{h\in\Zset}{\sup}\sum\limits_{r,s=-\infty}^{+\infty}|\gamma_j(s-h)\gamma_j(r-h)\gamma_j(r)\gamma_j(s)|
&\leq\underset{h}{\sup}\left(\sum\limits_{r=-\infty}^{+\infty}|\gamma_j(r)\gamma_j(r-h)|\right)^2
 \end{align*}
which is finite by the Cauchy-Schwarz inequality, since $\sum\limits_{r\in\Zset}\gamma^2_j(r)<\infty$.\\
 Using $|\gamma_j(h)|<\gamma_j(0)$ and the Cauchy-Schwarz inequality, we have
\begin{align*}
\underset{h\in\Zset}{\sup}\sum\limits_{r,s=-\infty}^{+\infty}|\gamma_j(h)\gamma_j(r-s)\gamma_j(h-r)\gamma_j(s)|\leq \gamma_j(0)\sum\limits_{u\in\Zset}\gamma^2_j(u)\sum\limits_{s\in\Zset}|\gamma_j(s)|,
\end{align*}
and the same bound applies to
\begin{align*} \underset{h\in\Zset}{\sup}\sum\limits_{r,s=-\infty}^{+\infty}|\gamma_j(h)\gamma_j(r-s)\gamma_j(h-s)\gamma_j(r)|.
\end{align*}
Hence,we have (\ref{cumulant}) by  (\ref{eq:DjProperties}), which achieves the proof of Step~1.

\noindent\textbf{Step~2}.
 Let us define
 \begin{equation}
 \label{eq:partial-sum}
 S_{n_j}(t)=\frac{1}{\sqrt{n_j}}\sum\limits_{i=1}^{\pent{n_jt}} (W_{j,i}^2-\sigma^2_j),
 \end{equation}
 where $\sigma_j^2=\PE(W_{j,i}^2)$, and $\pent{x}$ is the entire part of $x$.
Step~2 consists in proving that for $0\leq t_1\leq\ldots\leq t_k\leq1$, and $\mu_1,\ldots,\mu_k \in \Rset,$
\begin{equation} \label{fidi} \sum\limits_{i=1}^k\mu_iS_{n_j}(t_i)\cl\mathcal{N}\left(0,\frac{1}{\pi}\int_{-\pi}^{\pi}\den{j}\times
\PVar\left(\sum\limits_{i=1}^k\mu_iB(t_i)\right)\right).
\end{equation}
Observe that
\begin{align*}
\sum\limits_{i=1}^k\mu_iS_{n_j}(t_i)&=\frac{1}{\sqrt{n_j}}\sum\limits_{i=1}^k\mu_i\sum\limits_{l=1}^{n_j}(W_{j,l}^2-\sigma^2_j)\1_{\{l\leq\pent{n_jt_i}\}}\\
&= \sum\limits_{l=1}^{n_j}W_{j,l}^2a_{l,n}-E\left(\sum\limits_{l=1}^{n_j}W_{j,l}^2a_{l,n}\right)\\
&=\xi^T_{n_j}A_{n_j}\xi_{n_j}\;,
\end{align*}
where we set $a_{l,n}=\frac{1}{\sqrt{n_j}}\sum\limits_{i=1}^k\mu_i\1_{\{l\leq\pent{n_jt_i}\}}$, $\xi_{n_j}=\big(W_{j1},\ldots,W_{jn_{j}}\big)^T$ and $A_{n_j}$ is the diagonal matrix with diagonal entries $(a_{1,n_j},\ldots,a_{n_j,n_j})$.
Applying  \cite[Lemma~12]{moulines:roueff:taqqu:2008:aos}, (\ref{fidi}) is obtained by proving that, as $n_j\to\infty$,
\begin{align}\label{eq:rhorho}
&\rho(A_{n_j})\rho(\Gamma_{n_j})\to0 \\
\label{eq:varLimQuad}
&\PVar\left(\sum\limits_{i=1}^k\mu_iS_{n_j}(t_i)\right)\to\frac1\pi\int_{-\pi}^{\pi}\den{j}
\times \PVar\left(\sum\limits_{i=1}^k\mu_i\big(B(t_i)\big)\right)\;,
\end{align}
where $\rho(A)$ denote the spectral radius of the matrix $A$, that is, the maximum modulus of its eigenvalues and
$\Gamma_{n_j}$ is the covariance matrix of $\xi_{n_j}$.
The process $(W_{j,i})_{\{i=1,\dots,n_j\}}$ is stationary with spectral density $\bdens[\phi,\psi]{j,0}{.}{f}.$ Thus, by Lemma~2 in \cite{moulines:roueff:taqqu:2007:jtsa} its covariance matrix $\Gamma_{n_j}$ satisfies $\rho(\Gamma_{n_j})\leq 2\pi\underset{\lambda}{\sup}\bdens[\phi,\psi]{j,0}{\lambda}{f}.$
Furthermore, as $n_j\to\infty$,
\begin{align*} \rho(A_{n_j})=
\underset{1\leq l\leq n_j}{\max}\frac{1}{\sqrt{n_j}}\left|\sum\limits_{i=1}^k\mu_i\1_{\{l\leq\pent{n_jt_i}\}}\right|\leq n_j^{-1/2}\sum\limits_{i=1}^k|\mu_i|\rightarrow 0,
\end{align*}
and (\ref{eq:rhorho}) holds.
We now prove (\ref{eq:varLimQuad}). Using that $B(t)$ has variance $t$ and
independent and stationary increments, and that these properties characterize its covariance function, it is sufficient to show that,
for all $t\in[0,1]$, as $n_j\rightarrow \infty$,
\begin{equation}\label{eq:VarInt}
\PVar\left(S_{n_j}(t)\right)\to t\int_{-\pi}^{\pi}\den{j}\;,
\end{equation}
and for all $0\leq r\leq s\leq t\leq 1$, as $n_j\rightarrow \infty$,
\begin{equation}\label{eq:CovIntsr}
\cov\left(S_{n,j}(t)-S_{n,j}(s),S_{n,j}(r)\right)\to 0\;.
\end{equation}
For any sets $A,B\subseteq[0,1]$, we set
$$
V_{n_j}(\tau,A,B)=\frac{1}{n_j}\sum\limits_{k\geq1}\1_A((k+\tau)/n_j)\1_B(k/n_{j})\;.
$$
For all $0\leq s,t\leq 1$, we have
\begin{align*}
 \cov\left(S_{n_j}(t),S_{n_j}(s)\right)& =\frac{1}{n_j}\somme{i=1}{\pent{n_jt}}\somme{k=1}{\pent{n_{j}s}}\cov(W^2_{j,i},W^2_{j,k})\\
&= 2\sum\limits_{\tau\in\Zset}\gamma^2_{j}(\tau)V_{n_j}(
\tau,]0,t],]0,s])\;.
\end{align*}
The previous display applies to the left-hand side of (\ref{eq:VarInt}) when $s=t$ and for $0\leq r\leq s\leq t\leq 1$, it yields
$$
\cov\left(S_{n_j}(t)-S_{n_j}(s),S_{n_j}(r)\right)=2\sum\limits_{\tau\in\Zset}\gamma^2_{j}(\tau)V_{n_j}(\tau,]s,t],]0,r])\;.
$$
Observe that for all $A,B\subseteq[0,1]$,
$\underset{\tau}{\sup}\left|V_n(j,\tau,A,B\right|\leq \frac{k}{n_j}\leq 1$. Hence, by dominated convergence, the limits in (\ref{eq:VarInt}) and (\ref{eq:CovIntsr}) are obtained by computing the limits of $V_n(j,\tau,]0,t],]0,t])$ and $V_n(j,\tau,]s,t],]0,r])$ respectively.
We have  for any $\tau\in\Zset$, $t>0$, and $n_j$ large enough,
$$
\sum_{k\geq1}\1_{\{\frac{k+\tau}{n_j}\in]0,t]\}}\1_{\{\frac{k}{n_{j}}\in]0,t]\}}=\left\{(n_jt\wedge n_jt-\tau)\right\}_+=n_jt-\tau_+\;.
$$
Hence, as $n_j\to\infty$, $V_{n_j}(\tau,]0,t],]0,t])\to t$ and, by (\ref{autocov}), (\ref{eq:VarInt}) follows.
We have for any $\tau\in\Zset$ and $0< r\leq s \leq t$,
\begin{align*}
\sum\limits_{k\geq1}\1_{\{\frac{k+\tau}{n_j}\in]s,t]\}}\1_{\{\frac{k}{n_{j}}\in]0,r]\}}&=\left\{(n_jr\wedge \{n_jt-\tau\})-(0\vee \{n_js-\tau\})\right\}_+\\
&=(n_jr-n_js+\tau)_+ \to \1_{\{r=s\}} \; \tau_+ \;,
\end{align*}
where the last equality holds for $n_j$ large enough and the limit as $n_j\to\infty$. Hence $V_{n_j}(\tau,]s,t],]0,r])\to0$ and (\ref{eq:CovIntsr}) follows, which achieves Step~2.

\noindent\textbf{Step~3}. We now prove the tightness of $\{S_{n_j}(t),\;t\in[0,1]\}$ in the Skorokhod metric space. By Theorem 13.5 in \cite{billingsley:1999}, it is sufficient to prove that %
for all $0\leq r\leq s \leq t$,
$$
\PE\big[|S_{n_j}(s)-S_{n_j}(r)|^{2}|S_{n_j}(t)-S_{n_j}(s)|^{2}\big]\leq C |t-r|^{2}\;,
$$
where $C>0$ is some constant independent of $r$, $s$, $t$ and $n_j$.
We shall prove that, for all $0\leq r \leq t$,
\begin{equation}\label{eq:tightBiling}
\PE\big[|S_{n_j}(t)-S_{n_j}(r)|^{4}\big]\leq C_1 \{n_j^{-1}(\pent{n_jt}-\pent{n_jr})\}^{2}\;.
\end{equation}
By the Cauchy-Schwarz inequality, and using 
that, for $0\leq r\leq s \leq t$,
$$
n_j^{-1}(\pent{n_jt}-\pent{n_js})\times n_j^{-1}(\pent{n_js}-\pent{n_jr})\leq 4(t-r)^2 \;,
$$
the criterion (\ref{eq:tightBiling}) implies the previous criterion. Hence the tightness follows from (\ref{eq:tightBiling}), that we now prove.
We have, for any $\mathbf{i}=(i_1,\dots,i_4)$,
\begin{multline*}
\PE\left[\prod_{k=1}^4(W_{j,i_k}^2-\sigma_j^2)\right]=\Cum(W_{j,i_1}^2,\dots,W_{j,i_4}^2)+\PE(W_{j,i_1}^2,W_{j,i_2}^2)\PE[W_{j,i_3}^2,W_{j,i_4}^2]\\
\hspace{0.2cm}+\PE[W_{j,i_1}^2,W_{j,i_3}^2]\PE[W_{j,i_2}^2W_{i_4}^2]+\PE[W_{j,i_1}^2W_{i_4}^2]\PE[W_{j,i_2}^2W_{i_3}^2] \;.
\end{multline*}
It follows that, denoting for $0\leq r\leq t\leq 1$,
\begin{multline*}
\PE\left[\left|S_{n_j}(t)-S_{n_j}(r)\right|^4\right]=
\frac{1}{n_j^2}\sum_{\mathbf{i}\in A_{r,t}^4}
\Cum(W_{j,i_1}^2,\dots,W_{j,i_4}^2)\\
+\frac{3}{n_j^2}\left(\sum_{\mathbf{i}\in A_{r,t}^2}\PE[W_{j,i_1}W_{j,i_2}]\right)^2\;
\end{multline*}
where $A_{r,t}=\{\pent{n_jr}+1,\dots,\pent{n_jt}\}$.
Observe that
$$
0\leq \frac{1}{n_j}\sum_{\mathbf{i}\in A_{r,t}^2}\PE[W_{j,i_1}^2W_{j,i_2}^2]\leq 2\sum_{\tau\in\Zset}\gamma_j^2(\tau) \times n_j^{-1}(\pent{n_jt}-\pent{n_jr}) \;.
$$
Using that, by (\ref{eq:cum4}), $\Cum(W_{j,i_1}^2,\dots,W_{j,i_4}^2)=\cum(i_2-i_1,i_3-i_1,i_4-i_1)$, we have
\begin{align*}
\sum_{\mathbf{i}\in A_{r,t}^4}\left|\Cum\big(W_{j,i_1}^2,\dots,W_{j,i_4}^2\big)\right|
&\leq(\pent{n_jt}-\pent{n_jr})\sum_{h,s,l=\pent{n_jr}-\pent{n_jt}+1}^{\pent{n_jt}-\pent{n_jr}-1}\left|\cum(h,s,l)\right|\\
&\leq2(\pent{n_jt}-\pent{n_jr})^2\;\;\underset{h\in\Zset}{\sup}\sum\limits_{r,s=-\infty}^{+\infty}|\cum (h,r,s)| \;.
\end{align*}
The last three displays and (\ref{cumulant}) imply (\ref{eq:tightBiling}), which proves the tightness.

Finally, observing that the variance (\ref{autocov}) is positive, unless $f$ vanishes almost everywhere, the convergence (\ref{convergence}) follows from
Slutsky's lemma and the three previous steps.
\end{proof}
\subsection{The multiple-scale case}
\label{sec:multiple-scale}
The results above can  be extended to test simultaneously changes in wavelet variances occurring simultaneously at multiple time-scales.
To construct a multiple scale test, consider  the  \emph{between-scale} process
\begin{equation}\label{eq:betweenscaleProc}
\{[\dwt^X_{j,k},\,\bdwt^X_{j,k}(j-j')^T]^T \}_{k\in \Zset}\;,
\end{equation}
where the superscript $^T$ denotes the transpose and $\bdwt^X_{j,k}(\dj)$, $\dj=0,1,\dots,j$, is defined as follows:
\begin{equation}\label{eq:Defbd}
\bdwt^X_{j,k}(\dj) \eqdef \left[\dwt^X_{j-\dj,2^{\dj}k},\,\dwt^X_{j-\dj,2^{\dj}k+1},\,\dots,
  \dwt^X_{j-\dj,2^{\dj}k+2^{\dj}-1}\right]^T.
\end{equation}
It is a $2^\dj$-dimensional vector of wavelet coefficients at scale $j'=j-\dj$ and involves all
possible translations of the position index $2^{\dj}k$ by $\dk=0,1,\dots,2^{\dj}-1$.
The index $\dj$ in~(\ref{eq:Defbd}) denotes the scale difference $j-j'\geq0$ between the finest scale $j'$ and the coarsest
scale $j$. Observe that $\bdwt^X_{j,k}(0)$ ($\dj=0$) is the scalar $\dwt^X_{j,k}$.
It is shown in \cite[Corollary~1]{moulines:roueff:taqqu:2007:jtsa} that, when $\diffop^M X$ is covariance stationary, the between scale process $\{[\dwt^X_{j,k},\,\bdwt^X_{j,k}(j-j')^T]^T\}_{k\in \Zset}$ is also covariance stationary. Moreover, for all $0 \leq \dj \leq j$, the \emph{between scale covariance matrix} is
defined as
\begin{equation}\label{density}
\PCov\left(\dwt^X_{j,0},\bdwt^X_{j,k}(\dj) \right) = \int_{-\pi}^\pi \rme^{\rmi\lambda k} \,
\bdens[\phi,\psi]{j,\dj}{\lambda}{f} \, \rmd \lambda \; ,
\end{equation}
where $\bdens[\phi,\psi]{j,\dj}{\lambda}{f}$ is the cross-spectral density function of the between-scale process
given by (see~\cite[Corollary~1]{moulines:roueff:taqqu:2007:jtsa})
\begin{multline}
\label{eq:definitionfj}
\bdens[\phi,\psi]{j,\dj}{\lambda}{f} \eqdef
\sum_{l=0}^{2^{j}-1}  \be_{\dj}(\lambda+2l\pi)\, f(2^{-j}(\lambda+2l\pi))\,2^{-j/2} H_{j}(2^{-j}(\lambda+2l\pi))\\
\times 2^{-(j-\dj)/2}\overline{H_{j-\dj}(2^{-j}(\lambda+2l\pi))} \eqsp,
\end{multline}
where for all $\xi\in\Rset$,
$$
\be_\dj(\xi) \eqdef 2^{-\dj/2}\, [1, \rme^{-\rmi2^{-\dj}\xi}, \dots, \rme^{-\rmi(2^{\dj}-1)2^{-\dj}\xi}]^T\;.
$$
The case $\dj=0$ corresponds to the spectral density  of the \emph{within-scale} process $ \{ \dwt_{j,k} \}_{k \in \Zset}$
given in \eqref{density1scale}. Under the null hypothesis that $X$ is $K$-th order stationary, a \emph{multiple scale} procedure aims
at testing that the scalogram in a range satisfies
\begin{equation}
\label{eq:test1}
\mathcal{H}_0 : \sigma^2_{j,1}=\dots=\sigma^2_{j,n_j},  \text{for all} \;\; j\in \{J_1,J_1+1,\dots,J_2\}
\end{equation}
where $J_1$ and $J_2$ are the \textit{finest} and the \emph{coarsest} scales included in the procedure, respectively.
The wavelet coefficients at different scales are not uncorrelated so  that both the \emph{within-scale} and the
\emph{between scale} covariances need to be taken into account.
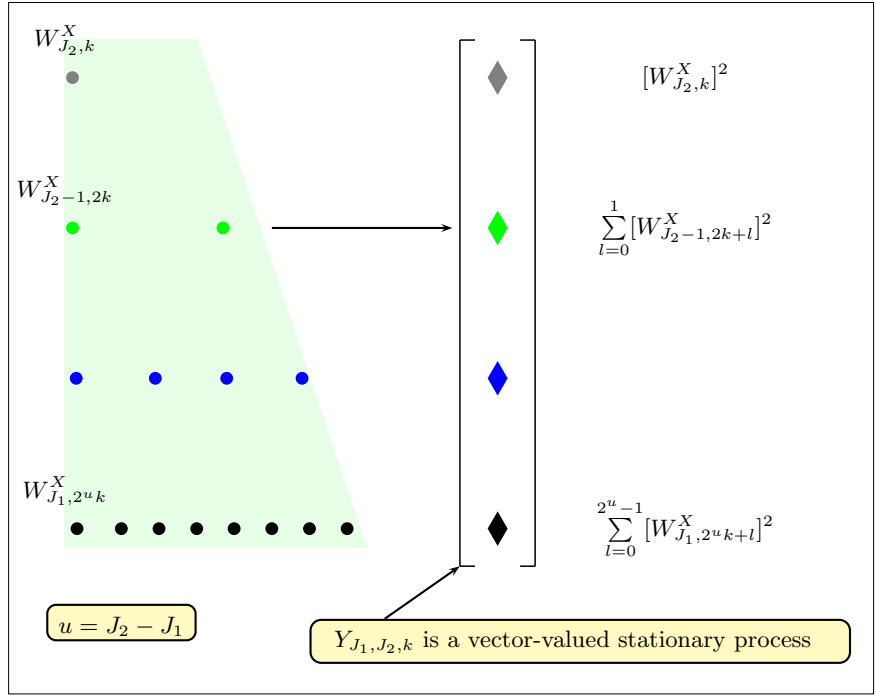
\begin{figure}[htbp]
\begin{pspicture}(0,-1.2)(11.5,8.2)
\psframe[linewidth=.001cm](0,-1.2)(11.5,8)
\pspolygon[fillstyle=solid,fillcolor=green!10!white,linecolor=green!10!white](0.75,0.75) (4.75,0.75) (2.5,7.5)(.75,7.5)
\psdots[linecolor=black,linewidth=.05cm](.91,1)(1.5,1)(2,1)(2.5,1)(3,1)(3.5,1)(4,1)(4.5,1)
\rput (.75,1.5){$W^X_{J_1,2^{u}k}$}
\rput (9,1){$\sum\limits_{l=0}^{2^{u}-1}[W^X_{J_1,2^{u}k+l}]^2$}
\psdots[linecolor=blue,linewidth=.05cm] (.9,3) (1.95,3)(2.9,3) (3.9,3)
\psdots [linecolor=green,linewidth=.05cm](.85,5) (2.85,5)
\rput (.75,5.5){$W^X_{J_2-1,2k}$}
\rput (9,5){$\sum\limits_{l=0}^1[W^X_{J_2-1,2k+l}]^2$}
\psdots [linecolor=gray,linewidth=.05cm](.85,7)
\rput (.75,7.5){$W^X_{J_2,k}$}
\rput (9,7){$[W^X_{J_2,k}]^2$}

\psline{->}(3.5,5)(5.9,5)
\psline[linewidth=.02] (6,0.5)(6,7.5)
\psline[linewidth=.02] (6,0.5) (6.2,0.5)
\psline[linewidth=.02] (6,7.5) (6.2,7.5)
\psline[linewidth=.02] (7,0.5)(7,7.5)
\psline[linewidth=.02] (7,0.5)(6.8,0.5)
\psline[linewidth=.02] (7,7.5)(6.8,7.5)

\psdots [dotstyle=diamond*,linecolor=gray,linewidth=.12](6.5,7)
\psdots [dotstyle=diamond*,linecolor=green,linewidth=.12](6.5,5)
\psdots [dotstyle=diamond*,linecolor=blue,linewidth=.12](6.5,3)
\psdots [dotstyle=diamond*,linecolor=black,linewidth=.12](6.5,1)
\psframe[framearc=.5,fillstyle=solid,fillcolor=yellow!30!white] (4,-.8)(11.2,-0.2)
\rput(7.5,-0.5){$Y_{J_1,J_2,k}$ is a vector-valued stationary process}
\psline{->}(5,-.2)(6,.5)
\psframe[framearc=.5,fillstyle=solid,fillcolor=yellow!30!white](.5,-.5)(2.5,0)
\rput(1.5,-.3){$\tiny{u=J_2-J_1}$}
\end{pspicture}
\caption{Between scale stationary process.}
\label{fig5}
\end{figure}

As before, we use a  CUSUM statistic in the wavelet domain. However, we now use multiple scale vector statistics.
Consider the following process
$$Y_{J_1,J_2,i}=\left(W^2_{J_2,i},\sum_{u=1}^{2}W^2_{J_2-1,2(i-1)+u},\dots,
\sum_{u=1}^{2^{(J_2-J_1)}}W^2_{J_1,2^{(J_2-J_1)}(i-1)+u}\right)^T \eqsp.$$
\textsl{The Bartlett estimator} of \textsl{the covariance matrix} of the square wavelet's coefficients for scales $\{J_1,\dots,J_2\}$ is the $(J_2-J_1+1)\times (J_2-J_1+1)$ symmetric definite
positive matrix $\hat{\Gamma}_{J_1,J_2}$ given by~:
\begin{align}\label{consistencyjj'}
&\hat{\Gamma}_{J_1,J_2}= \sum\limits_{\tau=-q(n_{J_2})}^{q(n_{J_2})}w_\tau[q(n_{J_2})] \hat{\gamma}_{J_1,J_2}(\tau) \eqsp, \;\; \text{where}\\
& \hat{\gamma}_{J_1,J_2}(\tau) = \frac{1}{n_{J_2}}  \sum\limits_{i,i+\tau=1}^{n_{J_2}}\left(Y_{J_1,J_2,i}-\bar{Y}_{J_1,J_2}\right)\left(Y_{J_1,J_2,i+\tau}-\bar{Y}_{J_1,J_2}\right)^T \eqsp.
\end{align}
where $\bar{Y}_{J_1,J_2}=\frac{1}{n_{J_2}}\sum\limits_{i=1}^{n_{J_2}}Y_{J_1,J_2,i}$ 

Finally, let us define  the vector of partial sum from scale $J_1$ to scale $J_2$ as
\begin{equation}\label{partial}
S_{J_1,J_2}(t)=\frac{1}{\sqrt{n_{J_2}}}\left[\sum_{i=1}^{\pent{n_jt}}W^2_{j,i}\right]_{j=J_1,\dots,J_2}\;.
\end{equation}
\begin{theorem} \label{theo:main-results-multiplescale}
  Under the assumptions of Theorem \ref{theo:main-result-singlescale}, we have, as $n\to\infty$,
\begin{equation}\label{eq:consistent}
\hat{\Gamma}_{J_1,J_2}=\Gamma_{J_1,J_2}+O_P\left(\frac{q(n_{J_2})}{n_{J_2}}\right)+O_P(q^{-1}(n_{J_2})),
\end{equation}
where $\Gamma_{J_1,J_2}(j,j')=\sum_{h\in\Zset}\cov(Y_{j,0},Y_{j',h})$, with $1\leq j,j'\leq J_2-J_1+1$ and,
\begin{equation}\label{multifidi}
\hat{\Gamma}_{J_1,J_2}^{-1/2} \left(S_{J_1,J_2}(t)-\PE\left[S_{J_1,J_2}(t)\right]\right) \cl B(t)
= \left(B_{J_1}(t),\ldots,B_{J_2}(t)\right),
\end{equation}
in $D^{J_2-J_1+1}[0,1]$, where $\{B_j(t)\}_{j=J_1,\ldots,J_2}$ are independent \textsl{Brownian motions}.
\end{theorem}

The proof of this result follows the same line as the proof of Theorem~\ref{theo:main-result-singlescale} and is therefore
omitted.
\section{Test statistics}
\label{sec:test-statistic}
Under the assumption of Theorem \ref{theo:main-result-singlescale}, the statistics
\begin{equation}
\label{eq:cusum-multiple}
 T_{J_1,J_2}(t)  \eqdef  \left(S_{J_1,J_2}(t)-tS_{J_1,J_2}(1)\right)^T \hat{\Gamma}_{J_1,J_2}^{-1}
\left(S_{J_1,J_2}(t)-tS_{J_1,J_2}(1)\right)
\end{equation}
converges in weakly in the  Skorokhod space $D([0,1])$
\begin{equation}
\label{eq:convergence-test}
 T_{J_1,J_2}(t) \cl \sum_{\ell=1}^{J_2-J_1-1} \left[ B_\ell^0(t) \right]^2
\end{equation}
where  $t \mapsto (B_1^0(t), \dots, B_{J_2-J_1+1}^0(t))$ is a vector of $J_2-J_1+1$ independent Brownian bridges

For any continuous function $F : D[0,1] \to  \Rset$, the continuous mapping Theorem implies that
\[
F[T_{J_1,J_2}(\cdot)] \cl F\left[ \sum_{\ell=1}^{J_2-J_1-1} \left[ B_\ell^0(\cdot) \right]^2 \right] \eqsp.
\]
We may for example apply either integral or max functionals, or weighted versions of these.
A classical example of integral function is  the so-called Cramér-Von Mises functional given by
\begin{equation}
\label{eq:definition-CVM}
\mathrm{CVM}(J_1,J_2) \eqdef \int_0^1 T_{J_1,J_2}(t) \rmd t \eqsp,
\end{equation}
which converges to $C(J_2-J_1+1)$ where for any integer $d$,
\begin{equation}
\label{eq:definition-C(d)}
C(d) \eqdef \int_0^1 \sum_{\ell=1}^d \left[B^0_\ell(t)\right]^2 \rmd t \eqsp.
\end{equation}
The test rejects the null hypothesis when
$\mathrm{CVM}_{J_1,J_2} \geq c(J_2-J_1+1,\alpha)$, where $c(d,\alpha)$ is the $1-\alpha$th quantile of the distribution of $C(d)$.
The distribution of the random variable $C(d)$
has been derived by \cite{kiefer:1959} (see also \cite{carmona:petit:pitman:yor:1999} for more recent references).
It holds that, for $x > 0$,
\[
\prob\left(C(d) \leq x \right) = \frac{2^{(d+1)/2}}{\pi^{1/2} x^{d/4}} \sum_{j=0}^\infty \frac{\Gamma(j+d/2)}{j! \Gamma(d/2)} \rme^{-(j+d/4)^2/x} \mathrm{Cyl}_{(d-2)/2} \left( \frac{2 j + d/2}{x^{1/2}} \right)
\]
where $\Gamma$ denotes the gamma function and $\mathrm{Cyl}$ are the parabolic cylinder functions.
The quantile of this distribution are given in table \ref{tab:quantile-CVM} for different values of $d=J_2-J_1+1$.
\begin{table}[t]
\centering
\begin{tabular}{|c||c|c|c|c|c|c|}\hline
Nominal S. & $d=1$ & $d=2$ & $d=3$ & $d=4$ & $d=5$ & $d=6$\\ \hline \hline
0.95 &0.4605& 0.7488 &1.0014& 1.2397 & 1.4691 & 1.6848 \\ \hline
0.99 &0.7401& 1.0721 &1.3521& 1.6267 & 1.8667 & 2.1259 \\\hline
\end{tabular}
\caption{Quantiles of the distribution $C(d)$ (see~(\ref{eq:definition-C(d)})) for different values of $d$}
\label{tab:quantile-CVM}
\end{table}
It is also possible to use the max. functional leading to an analogue of the Kolmogorov-Smirnov statistics,
\begin{equation}
\label{eq:definition-KSM}
\mathrm{KSM}(J_1,J_2) \eqdef  \sup_{0 \leq t \leq 1} T_{J_1,J_2}(t)
\end{equation}
which converges to $D(J_2-J_1+1)$ where for any integer $d$,
\begin{equation}
\label{eq:definition-D(d)}
D(d) \eqdef \sup_{0 \leq t \leq 1} \sum_{\ell=1}^d \left[ B_\ell^0(t) \right]^2 \eqsp.
\end{equation}
The test reject the null hypothesis when
$\mathrm{KSM}_{J_1,J_2} \geq \delta(J_2-J_1+1,\alpha)$, where $\delta(d,\alpha)$ is the $(1-\alpha)$-quantile of $D(d)$. The
distribution of $D(d)$ has again be derived by \cite{kiefer:1959} (see also \cite{pitman:yor:1999} for more recent references).
It holds that, for $x > 0$,
\[
\prob \left( D(d) \leq x \right)  = \frac{2^{1+(2-d)/2}}{ \Gamma(d/2) a^d} \sum_{n=1}^\infty \frac{j_{\nu,n}^{2\nu}}{J^2_{\nu+1}(j_{\nu,n})} \exp\left( - \frac{j_{\nu,n}^2}{2 x^2} \right) \eqsp,
\]
where $0 < j_{\nu,1} < j_{\nu,2} < \dots$ is the sequence of positive zeros of $J_\nu$, the Bessel function of index
$\nu = (d-2) /2$.
The quantiles of this distribution are given in Table \ref{tab:quantile-KSM}.
\begin{table}[t]
\centering
\begin{tabular}{|c||c|c|c|c|c|c|}\hline
$d$  & 1    & 2      &    3 &      4 &     5 &     6   \\ \hline
0.95 &1.358& 1.58379 &1.7472&  1.88226 & 2.00 & 2.10597 \\ \hline
0.99 &1.627624& 1.842726 &2.001& 2.132572 & 2.24798 & 2.35209 \\\hline
\end{tabular}
\caption{Quantiles of the distribution $D(d)$ (see~(\ref{eq:definition-D(d)}))
for different values of $d$.}
\label{tab:quantile-KSM}
\end{table}
\section{Power of the W2-CUSUM statistics}
\label{sec:power}
\subsection{Power of the test in single scale case}
In this section we investigate the power of the test. A minimal requirement is to establish that the test procedure
is pointwise consistent in a presence of a breakpoint, \ie\ that under a fixed alternative, the probability of detection converges to one as the sample size goes to infinity. We must therefore first define such alternative. For simplicity, we will consider an alternative where the process exhibit a single breakpoint, though it is likely that the test does have power against more general class of alternatives.

The alternative  that we consider in this section is defined as follows. Let $f_1$ and $f_2$ be two given generalized
spectral densities and suppose that, at a given scale $j$, $\int_{-\pi}^{\pi}|H_j(\lambda)|^2 f_i(\lambda) \rmd \lambda < \infty$, $i=1,2$, and
\begin{equation}\label{eq:H1}
\int_{-\pi}^{\pi}|H_j(\lambda)|^2\left(f_1(\lambda)-f_2(\lambda)\right) \rmd \lambda\neq 0 \eqsp.
\end{equation}
Define by $(X_{l,i})_{l\in\Zset}$, $i=1,2$, be two Gaussian processes, defined on the same probability space, with generalized
spectral density $f_1$. We do not specify the dependence structure between these two processes, which can be arbitrary.
Let $\kappa \in ]0,1[$ be a breakpoint. We consider a sequence of Gaussian processes $(X_k^n)_{k \in \Zset}$, such that
\begin{equation}
\label{eq:definition:X:alternative}
\text{$X_k^{(n)}= X_{k,i}$ for $k \leq \pent{n\kappa}$ and $X_k^{(n)} = X_{k,2}$ for $k \geq \pent{n\kappa} +1$}\eqsp.
\end{equation}
\begin{theorem}\label{theo:power-single-case}
Consider $\{ X_k^n \}_{k \in \Zset}$ be a sequence of processes specified by \eqref{eq:H1} and \eqref{eq:definition:X:alternative}.
Assume that $q(n_j)$ is non decreasing and :
\begin{equation}
q(n_j)\to \infty \;\; \text{and}\;\; \frac{q(n_j)}{n_j}\to 0 \;\text{as}\; n_j\to \infty \eqsp.
\end{equation}
Then the statistic $T_{n_j}$ defined by (\ref{eq:cusum}) satisfies
\begin{equation}
\frac{\sqrt{n_j}}{\sqrt{2q(n_j)}}\sqrt{\kappa(1-\kappa)}(1+o_p(1)) \leq T_{n_j} \cp\infty \eqsp.
\end{equation}
\end{theorem}
\begin{proof}
Let $k_j=\pent{n_j \kappa}$ the change point in the wavelet spectrum at scale $j$.
We write $q$ for $q(n_j)$ and suppress the dependence in $n$ in this proof to alleviate the notation. By
definition $T_{n_j} =\frac{1}{s_{j,n_j}} \underset{0 \leq t \leq 1}{\sup} \left( S_{n_j}(t) - t S_{n_j}(1) \right)$, where the process $t \mapsto S_{n_j}(t)$ is
defined in \eqref{eq:partial-sum}. Therefore, $T_{n_j} \geq \frac{1}{s_{j,n_j}}\left(S_{n_j}(\kappa) - \kappa S_{n_j}(1)\right)$. The proof consists in establishing that
$\frac{1}{s_{j,n_j}}\left(S_{n_j}(\kappa) - \kappa S_{n_j}(1)\right)= \frac{\sqrt{n_j}}{\sqrt{2q(n_j)}}\sqrt{\kappa(1-\kappa)}(1+o_p(1))$. We first decompose this
difference as follows
\begin{align*}
S_{n_j}(\kappa)- \kappa S_{n_j}(1)&=\frac{1}{\sqrt{n_j}}\left|\sum_{i=1}^{\pent{n_j \kappa}}W_{j,i}^2- \kappa \sum_{i=1}^{n_j}W_{j,i}^2\right|\\
&=B_{n_j}+f_{n_j}
\end{align*}
where $B_{n_j}$ is a fluctuation term
\begin{equation}\label{eq:Bnj} B_{n_j}=\frac{1}{\sqrt{n_j}}\left|\sum_{i=1}^{k_j}(W_{j,i}^2-\sigma_{j,i}^2)- \kappa \sum_{i=1}^{n_j}(W_{j,i}^2-\sigma_{j,i}^2)\right|
\end{equation}
and $f_{n_j}$ is a bias term
\begin{equation}\label{eq:fnj}
f_{n_j}=\frac{1}{\sqrt{n_j}}\left|\sum_{i=1}^{k_j}\sigma_{j,i}^2- \kappa \sum_{i=1}^{n_j}\sigma_{j,i}^2\right| \eqsp.
\end{equation}
Since support of $h_{j,l}$ is included in $[-\L(2^j+1),0]$ where $h_{j,l}$ is defined in (\ref{eq:FilterJ}), there exits a constant $a>0$ such that
\begin{align}\label{W1}
&W_{j,i}=W_{j,i;1}=\sum_{l\leq k}h_{j,2^ji-l}X_{l,1},\;\; \text{for}\;\; i<k_j,\\
\label{W2}
&W_{j,i}=W_{j,i;2}=\sum_{l>k}h_{j,2^ji-l}X_{l,2}\;\; \text{for}\;\; i>k_j+a,\\
&W_{j,i}=\sum_{l}h_{j,2^ji-l}X_l,\;\; \text{for}\;\; k_j\leq i<k_j+a.
\end{align}
Since the process $\{X_{l,1}\}_{l \in \Zset}$ and $\{ X_{l,2} \}_{l \in \Zset}$ are both $K$-th order covariance stationary,
the two processes $\{ W_{j,i;1} \}_{i \in \Zset}$ and $\{ W_{j,i;2} \}_{i \in \Zset}$ are also
covariance stationary. The wavelet coefficients $W_{j,i}$ for $i \in \{k_j, \dots, k_j+a\}$ are computed using observations from
the two processes $X_1$ and $X_2$.
Let us show that there exits a constant $C>0$ such that, for all integers $l$ and $\tau$,
\begin{equation}\label{var}
\PVar \left(\sum_{i=l}^{l+\tau}W_{j,i}^2 \right)\leq C\tau \eqsp.
\end{equation}
Using \eqref{autocov}, we have, for $\epsilon=1,2$,
$$
\PVar\left(\sum_{i=l}^{l+\tau}  W_{j,i;\epsilon}^2 \right)\leq
\frac{\tau}{\pi}\int_{-\pi}^{\pi}\left|\mathbf{D}_{j,0;\epsilon}(\lambda)\right|^2 \rmd \lambda
$$
where, $\mathbf{D}_{j,0;1}(\lambda)$ and $\mathbf{D}_{j,0;2}(\lambda)$ denote the spectral density of
the stationary processes $\{ W_{j,i;1} \}_{i \in \Zset}$ and $\{ W_{j,i;2} \}_{i \in \Zset}$ respectively. Using Minkovski inequality, we have for $l\leq k_j\leq k_j+a<l+\tau$ that
$\left(\PVar\sum_{i=l}^{l+\tau}W_{j,i}^2\right)^{1/2}$ is at most
\begin{align*}
&\left(\PVar\sum_{i=l}^{k_j}W_{j,i}^2\right)^{1/2}+\sum_{i=k_j+1}^{k_j+a}\left(\PVar
  W_{j,i}^2\right)^{1/2}+\left(\PVar\sum_{i=k_j+a+1}^{l+\tau}W_{j,i}^2\right)^{1/2}\\
&\hspace{0.5cm}\leq \left(\PVar\sum_{i=l}^{k_j}W_{j,i;1}^2\right)^{1/2}+a\sup_i(\PVar W_{j,i}^2)^{1/2}+\left(\PVar\sum_{i=k_j+a+1}^{l+\tau}W_{j,i;2}^2 \right)^{1/2}.
\end{align*}
Observe that $\PVar(W_{j,i}^2)\leq 2(\sum_{l}|h_{j,l}|)^2\left(\sigma_{j,1}^2\vee \sigma_{j,2}^2\right)^2<\infty$ for $ k_j\leq i<k_j+a,$ where
\begin{equation}\label{eqvar}
 \sigma_{j;1}^2=\PE\left[W_{j,i;1}^2\right],\;\, \text{and}\;\, \sigma_{j;2}^2=\PE\left[W_{j,i;2}^2\right]
\end{equation}
The three last displays imply (\ref{var}) and thus that $B_{n_j}$ is bounded in probability. Moreover, since $f_{n_j}$ reads
\begin{align*}
&\frac{1}{\sqrt{n_j}}\left|\sum_{i=1}^{\pent{n_j \kappa}} \sigma_{j;1}^2-\kappa\sum_{i=1}^{\pent{n_j\kappa}} \sigma_{j;1}^2-\kappa\sum_{i=\pent{n_j \kappa}+1}^{\pent{n_j \kappa}+a}\sigma_{j,i}^2-\kappa\sum_{i=\pent{n_j \kappa}+a+1}^{n_j} \sigma_{j;2}^2\right|\\
&\hspace{4cm}= \sqrt{n_j}\kappa(1-\kappa)\left|\sigma_{j;1}^2-\sigma_{j;2}^2\right|+O(n_j^{-1/2}) \eqsp,
\end{align*}
 we get
\begin{equation}\label{power}
S_{n_j}(\kappa)-\kappa S_{n_j}(1)=\sqrt{n_j}\kappa(1-\kappa)\left(\sigma_{j;1}^2- \sigma_{j;2}^2\right)+O_P(1) \eqsp.
\end{equation}
We now study the denominator $s_{j,n_j}^2$ in \eqref{eq:cusum}.
Denote by
\[
\bar{\sigma}_j^2=\frac{1}{n_j}\sum_{i=1}^{n_j}\sigma_{j,i}^2
\]
the expectation of the scalogram (which now differs from the wavelet spectrum). Let us consider for $\tau\in\{0,\dots,q(n_j)\}$ $\hat{\gamma}_j(\tau)$ the empirical covariance of the wavelet coefficients defined in \eqref{eq:gamma}.
\begin{multline*}
\hat{\gamma}_j(\tau)=\frac{1}{n_j}\sum_{i=1}^{n_j-\tau}(W_{j,i}^2-\bar{\sigma}_j^2)(W_{j,i+\tau}^2-\bar{\sigma}_j^2)-(1+\frac{\tau}{n_j})\left(\bar{\sigma}_j^2-\scalogram{j}\right)^2\\
+\frac{1}{n_j}(\scalogram{j}-\bar{\sigma}_j^2)\left\{\sum_{i=n_j-\tau+1}^{n_j}(W_{j,i}^2-\bar{\sigma}_j^2)+\sum_{i=1}^{\tau}(W_{j,i}^2-\bar{\sigma}_j^2)\right\}.
\end{multline*}
Using Minkowski inequality and (\ref{var}), there exists a constant $C$ such that for all $1\leq l\leq l+\tau\leq n_j,$
\begin{align*}
\left\|\sum_{i=l}^{l+\tau}\left(W_{j,i}^2-\Sbar\right)\right\|_2&\leq\left\|\sum_{i=l}^{l+\tau}\left(W_{j,i}^2-\sigma_{j,i}^2\right)\right\|_2+\left\|\sum_{i=l}^{l+\tau}\left(\sigma_{j,i}^2-\Sbar\right)\right\|_2\\
&\leq C(\tau^{1/2}+\tau),
\end{align*}
and similarly $$\left\|\scalogram{j}-\bar{\sigma}_j^2\right\|_2\leq\frac{C}{\sqrt{n_j}}.$$
By combining these two latter bounds, the Cauchy-Schwarz inequality implies that
\begin{equation*}\label{eq:esperance}
\left\|\frac{1}{n_j}\left(\scalogram{j}-\bar{\sigma}_j^2\right)\sum_{i=l}^{l+\tau}(W_{j,i}^2-\bar{\sigma}_j^2)\right\|_1\leq \frac{C(\tau^{1/2}+\tau)}{n_j^{3/2}} \eqsp.
\end{equation*}
Recall that $s_{j,n_j}^2=\sum_{\tau=-q}^{q}w_{\tau}(q)\hat{\gamma}_j(\tau)$ where $w_{\tau}(q)$ are the so-called Bartlett weights
defined in \eqref{eq-bandwidth}.
We now use the bounds above to identify the limit of $s_{j,n_j}^2$ as the sample size goes to infinity. The two previous identities imply that
\begin{align*}
\sum_{\tau=0}^{q}w_{\tau}(q) \left(1+\frac{\tau}{n_j} \right)\left\|\bar{\sigma}_j^2-\scalogram{j}\right\|_2& \leq C\frac{q^{2}}{n_j^{3/2}}
\end{align*}
and
\begin{align*}
\sum_{\tau=0}^{q}w_{\tau}(q)\left\|\frac{1}{n_j}\left(\scalogram{j}-\bar{\sigma}_j^2\right)\sum_{i=l}^{l+\tau}(W_{j,i}^2-\bar{\sigma}_j^2)\right\|_1&\leq C\frac{q^{2}}{n_j^{3/2}},
\end{align*}
Therefore, we obtain
\begin{equation}
s_{j,n_j}^2=\sum_{\tau=-q}^{q}w_{\tau}(q)\tilde{\gamma_j}(\tau)+O_P\left(\frac{q^2}{n_j^{3/2}}\right),
\end{equation}
 where $\tilde{\gamma}_j(\tau)$ is defined by
\begin{equation}\label{eq:main} \tilde{\gamma}_j(\tau)=\frac{1}{n_j}\sum_{i=1}^{n_j-\tau}(W_{j,i}^2-\bar{\sigma}_j^2)(W_{j,i+\tau}^2-\bar{\sigma}_j^2).
\end{equation}
Observe that since $q=o(n_j)$, $k_j=\pent{n_j \kappa}$ and $0\leq\tau\leq q$, then for any given integer $a$ and $n$ large enough $0\leq\tau\leq k_j
\leq k_j+a\leq n_j-\tau$ thus in (\ref{eq:main}) we may write
$\sum_{i=1}^{n_j-\tau}=\sum_{i=1}^{k_j-\tau}+\sum_{i=k_j-\tau+1}^{k+a}+\sum_{i=k_j+a+1}^{n_j-\tau}$.
Using $\sigma_{j;1}^2$ and $\sigma_{j;2}^2$ in (\ref{eq:main}) and
straightforward bounds that essentially follow from~(\ref{var}), we get $s_{j,n_j}^2= \bar{s}_{j,n_j}^2 +O_P\left(\frac{q^2}{n_j}\right)$, where
\begin{multline*}
\bar{s}_{j,n_j}^2=\sum_{\tau=-q}^{q}w_{\tau}(q)\Bigg(\frac{k}{n_j}\tilde{\gamma}_{j;1}(\tau)+\frac{n_j-k_j-a}{n_j}\tilde{\gamma}_{j;2}(\tau)\\
+\frac{k_j-|\tau|}{n_j}\left(\sigma_{j;1}^2-\bar{\sigma}_j^2\right)^2
+\frac{n_j-k_j-a-|\tau|}{n_j}\left(\sigma_{j;2}^2-\bar{\sigma}_j^2\right)^2\Bigg)
\end{multline*}
with
\begin{align*}
&\tilde{\gamma}_{j;1}(\tau)=\frac{1}{k_j}\sum_{i=1}^{k_j-\tau}\left(W_{j,i}^2-\sigma_{j;1}^2\right)\left(W_{j,i+\tau}^2-\sigma_{j;1}^2\right), \\
&\tilde{\gamma}_{j;2}(\tau)=\frac{1}{n_j-k_j-a}\sum_{i=k_j+a+1}^{n_j-\tau}\left(W_{j,i}^2-\sigma_{j;2}^2\right)\left(W_{j,i+\tau}^2-\sigma_{j;2}^2\right) \;.
\end{align*}
Using that $\bar{\sigma}_j^2\to \kappa \sigma_{j;1}^2+(1-\kappa) \sigma_{j;2}^2$ as $n_j\to
\infty,$ and that,  for $\epsilon=1,2$,
\[
s_{j,n_j;\epsilon}^2\eqdef\sum\limits_{\tau=-q}^{q}w_{\tau}(q)\tilde{\gamma}_{j;\epsilon}(\tau)\cp
\frac1\pi\int_{-\pi}^{\pi}|\mathbf{D}_{j,0;\epsilon}(\lambda)|^2 \rmd\lambda  \eqsp,
\]
we obtain
\begin{multline}\label{eq:s2asympH1}
s_{j,n_j}^2=
\frac{1}{\pi}\int_{-\pi}^{\pi}\left\{\kappa\left|\mathbf{D}_{j,0;1}(\lambda)\right|^2+ (1-\kappa)\left|\mathbf{D}_{j,0;2}(\lambda)\right|^2\right\} \rmd \lambda\\
2q \kappa(1-\kappa)\left(\sigma_{j;1}^2-\sigma_{j;2}^2\right)^2+ +o_p(1)+O_P\left(\frac{q^2}{n_j}\right).
\end{multline}
Using (\ref{power}), the last display and that $o_p(1)+O_P\left(\frac{q^2}{n_j}\right)=o_p(q),$ we finally obtain
\begin{align*}
S_{n_j}(\kappa) - \kappa S_{n_j}(1)&=\frac{\sqrt{n_j}\kappa(1-\kappa)\left|\sigma_{j;1}^2-\sigma_{j;2}^2\right|+O_P(1)}{\sqrt{2q(\kappa(1-\kappa))}\left|\sigma_{j;1}^2-\sigma_{j;2}^2\right|+o_p(\sqrt{q})}\\
&=\frac{\sqrt{n_j}}{\sqrt{2q}}\sqrt{\kappa(1-\kappa)}(1+o_p(1)) \eqsp,
\end{align*}
which concludes the proof of Theorem~\ref{theo:power-single-case}.
\end{proof}
\subsection{Power of the test in multiple scales case}
The results obtained in the previous Section in the single scale case easily extend to the test procedure designed to handle the multiple
scales case. The alternative is specified exactly in the same way than in the single scale case but instead of considering the square of the wavelet coefficients at a given scale, we now study the behavior of the between-scale process.
Consider the following process for $\epsilon=1,2$,
$$
Y_{J_1,J_2,i;\epsilon}=\left(W^{2}_{J_2,i;\epsilon},\sum_{u=1}^{2}W^{2}_{J_2-1,2(i-1)+u;\epsilon},\dots,
\sum_{u=1}^{2^{(J_2-J_1)}}W^{2}_{J_1,2^{(J_2-J_1)}(i-1)+u; \epsilon}\right)^T\;,
$$
where $J_1$ and $J_2$ are respectively the finest and the coarsest scale considered in the test, $W_{j,i;\epsilon}$ are defined in (\ref{W1}) and (\ref{W2}) and $\Gamma_{J_1,J_2;\epsilon}$  the $(J_2-J_1+1)\times (J_2-J_1+1)$ symmetric non negative matrix such that
\begin{equation}\label{mat} \Gamma_{J_1,J_2;\epsilon}(j,j')=\sum\limits_{h\in\Zset}\cov(Y_{j,0;\epsilon},Y_{j',h;\epsilon})=\int_{-\pi}^{\pi}\left\|\mathbf{D}_{j,u;\epsilon}(\lambda;f)\right\|^2d\lambda,
\end{equation}
with $1\leq j,j'\leq J_2-J_1+1$ for $\epsilon=1,2$.
\begin{theorem}\label{theo:power-multi-case}
Consider $\{ X_k^n \}_{k \in \Zset}$ be a sequence of processes specified by \eqref{eq:H1} and \eqref{eq:definition:X:alternative}.
Finally assume that for at least one $j\in\{J_1,\dots,J_2\}$  and that at least one of the two matrices $\Gamma_{J_1,J_2;\epsilon}$ $\epsilon=1,2$ defined in (\ref{mat}) is positive definite. Assume in addition that
Finally, assume that the number of lags $q(n_{J_2})$ in the Barlett estimate of the covariance matrix \eqref{consistencyjj'} is non decreasing and:
\begin{equation}\label{eq:window}
q(n_{j_2})\to \infty  \quad \text{and} \quad \frac{q^2(n_{J_2})}{n_{J_2}}\to 0,\;\,\text{as}\;\,n_{J_2}\to\infty,\; \eqsp.
\end{equation}
Then, the W2-CUSUM test statistics  $T_{J_1,J_2}$ defined by (\ref{eq:cusum-multiple}) satisfies
\begin{align*}
\frac{n_{J_2}}{2q(n_{J_2})}\kappa(1-\kappa)\left(1+o_p(1)\right) \leq T_{J_1,J_2} \cp\infty \;\text{as}\; \quad n_{J_2}\to \infty
\end{align*}
\end{theorem}
\begin{proof}
As in the single scale case we drop the dependence in $n_{J_2}$ in the expression of $q$ in this proof section.
Let $k_j=\pent{n_{j} \kappa}$ the change point in the wavelet spectrum at scale $j.$ Then using (\ref{partial}), we have
that $T_{J_1,J_2} \geq S_{J_1,J_2}(\kappa)- \kappa S_{J_1,J_2}(1)$ where
$$
S_{J_1,J_2}(\kappa)- \kappa S_{J_1,J_2}(1)=\frac{1}{\sqrt{n_{J_2}}}\left[n_j(B_{n_j}+f_{n_j})\right]_{j=J_1,\dots,J_2}\;,
$$
where $B_{n_j}$ and $f_{n_j}$ are defined respectively by (\ref{eq:Bnj}) and (\ref{eq:fnj}). Hence as in (\ref{power}), we have
$$
S_{J_1,J_2}(\kappa)-\kappa S_{J_1,J_2}(1)=\sqrt{n_{J_2}}\kappa(1-\kappa)\Delta+O_P(1)\;,
$$
where $\Delta=\left[\sigma_{J_1,J_2;1}^2-\sigma_{J_1,J_2;2}^2\right]^T$ and $$\sigma_{J_1,J_2;\epsilon}^2=\left(\sigma_{J_2;\epsilon}^2,\dots,2^{J_2-J_1} \sigma_{J_1;\epsilon}^2\right)^T \eqsp.$$
We now study the asymptotic behavior of $\hat{\Gamma}_{J_1,J_2}$.
Using similar arguments as those leading to \eqref{eq:s2asympH1} in the proof of Theorem~\ref{theo:power-single-case}, we have
\begin{multline*}
\hat{\Gamma}_{J_1,J_2}=2q \kappa (1- \kappa)\Delta\Delta^T
+ \kappa \Gamma_{J_1,J_2;1}+(1-\kappa)\Gamma_{J_1,J_2;2}\\+O_P\left(\frac{q}{n_{J_2}}\right)+O_P\left(q^{-1}\right)+O_P\left(\frac{q^2}{n_{J_2}}\right).
\end{multline*}
For $\Gamma$ a positive definite matrix, consider the matrix $\textbf{M}(\Gamma)=\Gamma+2q \kappa(1-\kappa)\Delta\Delta^T$.
Using the matrix inversion lemma, the inverse of $\textbf{M}(\Gamma)$ may be expressed as
$$\textbf{M}^{-1}(\Gamma)=\left(\Gamma ^{-1}-\frac{2q \kappa(1-\kappa)\Gamma ^{-1}\Delta\Delta^T\Gamma^{-1}}{1+2q \kappa(1-\kappa)\Delta^T\Gamma^{-1}\Delta}\right) \eqsp,$$
which implies that
$$\Delta^T\textbf{M}^{-1}(\Gamma)\Delta=\frac{\Delta^T\Gamma^{-1}\Delta}{1+2q \kappa(1-\kappa)\Delta^T\Gamma^{-1}\Delta}.$$
Applying these two last relations to $\Gamma_0= \kappa \Gamma_{J_1,J_2}^{(1)}+(1-\kappa)\Gamma_{J_1,J_2}^{(2)}$ which is symmetric and definite positive (since, under the stated assumptions at least one of the two matrix  $\Gamma_{J_1,J_2;\epsilon}$, $\epsilon=1,2$ is positive)
we have
\begin{align*}
T_{J_1,J_2} &\geq \kappa^2(1-\kappa)^2n_{J_2}\Delta^T\textbf{M}^{-1}\left(\Gamma_0+O_P\left(\frac{q^2}{n_{J_2}}\right)+O_P(q^{-1})\right)\Delta+O_P(1)\\
&=n_{J_2}\kappa^2(1-\kappa)^2\frac{\Delta^T\Gamma_0^{-1}\Delta+O_P\left(\frac{q^2}{n_{J_2}}\right)+O_P(q^{-1}}{2q\kappa(1-\kappa)\Delta^T\Gamma_0^{-1}\Delta(1+o_p(1))}+O_P(1)\\
&=\frac{n_{J_2}}{2q}\kappa(1-\kappa)\left(1+o_p(1)\right)\;.
\end{align*}
Thus $T_{J_1,J_2}\cp\infty$ as $n_{J_2}\to \infty$,  which completes the proof of Theorem~\ref{theo:power-multi-case}.
\end{proof}

\begin{remark}
The term corresponding to the "bias" term $\kappa \Gamma_{J_1,J_2;1}+(1-\kappa)\Gamma_{J_1,J_2;2}$ in the single case is
$\frac{1}{\pi}\int_{-\pi}^{\pi}\{\kappa|\mathbf{D}_{j,0;1}(\lambda)|^2+
(1-\kappa)|\mathbf{D}_{j,0;2}(\lambda)|^2\} \rmd \lambda=O(1)$, which can be neglected since the main term in $s_{j,n_j}^2$ is of
order $q\to\infty$. In multiple scale case, the main term in $\hat{\Gamma}_{J_1,J_2}$ is still of order $q$ but is no longer
invertible (the rank of the leading term is equal to 1). A closer look is thus necessary and the term $\kappa \Gamma_{J_1,J_2;1}+(1-\kappa)\Gamma_{J_1,J_2;2}$ has to be taken into account. This is also explains why we need the more stringent condition \eqref{eq:window} on the bandwidth size in the multiple scales case.
\end{remark}

\section{Some examples}
\label{sec:applications}
 In this section, we report the results of a limited Monte-Carlo experiment to assess the finite sample property of the test procedure.
 Recall that the test rejects the null if either $\mathrm{CVM}(J_1,J_2)$ or $\mathrm{KSM}(J_1,J_2)$, defined in \eqref{eq:definition-CVM} and
 \eqref{eq:definition-KSM} exceeds the $(1-\alpha)$-th quantile of the distributions $C(J_2-J_1+1)$ and $D(J_2-J_1+1)$, specified in
 \eqref{eq:definition-C(d)} and \eqref{eq:definition-D(d)}. The quantiles are reported in Tables \eqref{tab:quantile-CVM} and
 \eqref{tab:quantile-KSM}, and have been obtained by truncating the series expansion of the cumulative distribution function.
To study the influence on the test procedure of the strength of the dependency, we consider different classes of Gaussian processes, including white noise, autoregressive moving average (ARMA) processes as well as fractionally integrated ARMA (ARFIMA($p,d,q$)) processes which are
 known to be long range dependent. In all the simulations we set the lowest scale to $J_1=1$ and vary the coarsest scale $J_2=J$. We used a
 wide range of values of sample size  $n$, of the number of scales $J$  and of the parameters of the ARMA and FARIMA
 processes but, to conserve space, we present the results only for $n=512,1024,2048,4096,8192$, $J=3,4,5$ and four different
 models: an AR(1) process with parameter $0.9$, a MA(1) process with parameter 0.9, and two ARFIMA(1,d,1) processes with
 memory parameter $d=0.3$ and $d=0.4$, and the same AR and MA coefficients, set to 0.9 and 0.1.
 In our simulations, we have used the Newey-West estimate of the bandwidth $q(n_j)$ for the covariance estimator
 (as implemented in the R-package \emph{sandwich}).
\subsubsection{Asymptotic level of $KSM$ and $CVM$.}
We investigate the finite-sample
behavior of the test statistics $\mathrm{CVM}(J_1,J_2)$ and $\mathrm{KSM}(J_1,J_2)$ by computing the number of times that the null hypothesis is rejected
in $1000$ independent replications of each of these processes under $\mathcal{H}_0$ , when the asymptotic level is set to $0.05$.
\begin{table}\centering
 \begin{tabular}{*{7}{c}}

        \multicolumn{1}{c}{} & \multicolumn{5}{c}{{White noise}}&\multicolumn{1}{c}{} \\
        \hline\hline
        \multicolumn{2}{c}{$n$}& 512& 1024 & 2048& 4096 & 8192  \\
        \hline
        $J=3$& $KSM$ & {0.02}&{ 0.01}& {0.03}&{ 0.02}& {0.02}\\
        $J=3$& $CVM$ & {0.05}&{ 0.045}& {0.033}&{ 0.02}& {0.02}\\
        \hline
        $J=4$ &$KSM$ & {0.047}&{ 0.04}& {0.04}&{ 0.02}& {0.02}\\
        $J=4$&$CVM$ & {0.041}&{ 0.02}& {0.016}&{ 0.016}& {0.01}\\
        \hline
        $J=5$ &$KSM$ & {0.09}&{ 0.031}& {0.02}&{ 0.025}& {0.02}\\
        $J=5$&$CVM$ & {0.086}&{ 0.024}& {0.012}&{ 0.012}& {0.02}\\
        \hline
        \hline

        \end{tabular}
\caption{Empirical level of $\texttt{KSM}-\texttt{CVM}$ for a white noise.}
\label{tab1}
\end{table}
 \begin{table}\centering
 \begin{tabular}{*{7}{c}}

        \multicolumn{1}{c}{} & \multicolumn{5}{c}{{MA(1)$[\theta=0.9]$}}&\multicolumn{1}{c}{} \\
        \hline\hline
        \multicolumn{2}{c}{$n$}& 512& 1024 & 2048& 4096 & 8192  \\
        \hline
        $J=3$& $KSM$ & {0.028}&{ 0.012}& {0.012}&{ 0.012}& {0.02}\\
        $J=3$& $CVM$ & {0.029}&{ 0.02}& {0.016}&{ 0.016}& {0.01}\\
        \hline
        $J=4$ &$KSM$ & {0.055}&{ 0.032}& {0.05}&{ 0.025}& {0.02}\\
        $J=4$&$CVM$ & {0.05}&{ 0.05}& {0.03}&{ 0.02}& {0.02}\\
        \hline
        $J=5$ &$KSM$ & {0.17}&{ 0.068}& {0.02}&{ 0.02}& {0.02}\\
        $J=5$&$CVM$ & {0.13}&{ 0.052}& {0.026}&{ 0.021}& {0.02}\\
       \hline
       \hline

        \end{tabular}
\caption{Empirical level of $\texttt{KSM}-\texttt{CVM}$ for a $MA(q)$ process.}
\label{tab2}
\end{table}

 \begin{table}\centering
 \begin{tabular}{*{7}{c}}

        \multicolumn{1}{c}{} & \multicolumn{5}{c}{{AR(1)$[\phi=0.9]$}}&\multicolumn{1}{c}{} \\
        \hline\hline
        \multicolumn{2}{c}{$n$}& 512& 1024 & 2048& 4096 & 8192  \\
        \hline
        $J=3$& $KSM$ & {0.083}&{ 0.073}& {0.072}&{ 0.051}& {0.04}\\
        $J=3$& $CVM$ & {0.05}&{ 0.05}& {0.043}&{ 0.032}& {0.03}\\
        \hline
        $J=4$ &$KSM$ & {0.26}&{ 0.134}& {0.1}&{ 0.082}& {0.073}\\
        $J=4$&$CVM$ & {0.14}&{ 0.092}& {0.062}&{ 0.04}& {0.038}\\
        \hline
        $J=5$ &$KSM$ & {0.547}&{ 0.314}& {0.254}&{ 0.22}& {0.11}\\
        $J=5$&$CVM$ & {0.378}&{ 0.221}& {0.162}&{ 0.14}& {0.093}\\

        \hline\hline

        \end{tabular}
\caption{Empirical level of $\texttt{KSM}-\texttt{CVM}$ for an
         $AR(1)$ process.}
\label{tab3}
\end{table}

 \begin{table}\centering
 \begin{tabular}{*{7}{c}}
        \multicolumn{1}{c}{} & \multicolumn{5}{c}{{ARFIMA(1,0.3,1)$[\phi=0.9,\theta=0.1]$}}&\multicolumn{1}{c}{} \\
        \hline \hline
        \multicolumn{2}{c}{$n$}& 512& 1024 & 2048& 4096 & 8192  \\
        \hline
        $J=3$& $KSM$ & {0.068}&{ 0.047}& {0.024}&{ 0.021}& {0.02}\\
        $J=3$& $CVM$ & {0.05}&{ 0.038}& {0.03}&{ 0.02}& {0.02}\\
        \hline
        $J=4$ &$KSM$ & {0.45}&{ 0.42}& {0.31}&{ 0.172}& {0.098}\\
        $J=4$&$CVM$ & {0.39}&{ 0.32}& {0.20}&{ 0.11}& {0.061}\\
        \hline
        $J=5$& $KSM$ & {0.57}&{ 0.42}& {0.349}&{ 0.229}& {0.2}\\
        $J=5$& $CVM$ & {0.41}&{ 0.352}& {0.192}&{ 0.16}& {0.11}\\

        \hline \hline

        \end{tabular}
\caption{Empirical level of $\texttt{KSM}-\texttt{CVM}$ for an $ARFIMA(1,0.3,1)$ process.}
\label{tab4}
\end{table}

 \begin{table}\centering
 \begin{tabular}{*{7}{c}}
        \multicolumn{1}{c}{} & \multicolumn{5}{c}{{ARFIMA(1,0.4,1)$[\phi=0.9,\theta=0.1]$}}&\multicolumn{1}{c}{} \\
        \hline \hline
         \multicolumn{2}{c}{$n$}& 512& 1024 & 2048& 4096 & 8192  \\

        \hline
        $J=3$& $KSM$ & {0.11}&{ 0.063}& {0.058}&{ 0.044}& {0.031}\\
        $J=3$& $ CVM$ & {0.065}&{ 0.05}& {0.043}&{ 0.028}& {0.02}\\
        \hline
        $J=4$ &$KSM$ & {0.512}&{ 0.322}& {0.26}&{ 0.2}& {0.18}\\
        $J=4$&$CVM$ & {0.49}&{ 0.2}& {0.192}&{ 0.16}& {0.08}\\
        \hline
        $J=5$ &$KSM$ & {0.7}&{ 0.514}& {0.4}&{ 0.321}& {0.214}\\
        $J=5$&$CVM$ & {0.59}&{ 0.29}& {0.262}&{ 0.196}& {0.121}\\

        \hline\hline

        \end{tabular}
\caption{Empirical level of $\texttt{KSM}-\texttt{CVM}$ for an $ARFIMA(1,0.3,1)$ process.}
\label{tab5}
\end{table}
\newpage
We notice that in general the empirical levels for the CVM are globally more accurate than the ones for the KSM test, the difference being more significant when the
 strength of the dependence is increased, or when the number of scales that are tested simultaneously get larger. The tests are slightly too conservative in the white noise and the MA case (tables
 \eqref{tab1} and \eqref{tab2}); in the AR(1) case and in the ARFIMA cases, the test rejects the null much too often when the number of scales is large compared to the sample size
 (the difficult problem being in that case to estimate the covariance matrix of the test). For $J=4$, the number of samples required to meet the target rejection rate can be as large as
 $n=4096$ for the CVM test and $n=8192$ for the KSM test. The situation is even worse in the ARFIMA case (tables \eqref{tab4} and \eqref{tab5}).
 When the number of scales is equal to $4$ or $5$, the test rejects the null hypothesis
 much too often.

\begin{figure}[htbp]
\centerline{\includegraphics[width=.85\linewidth]{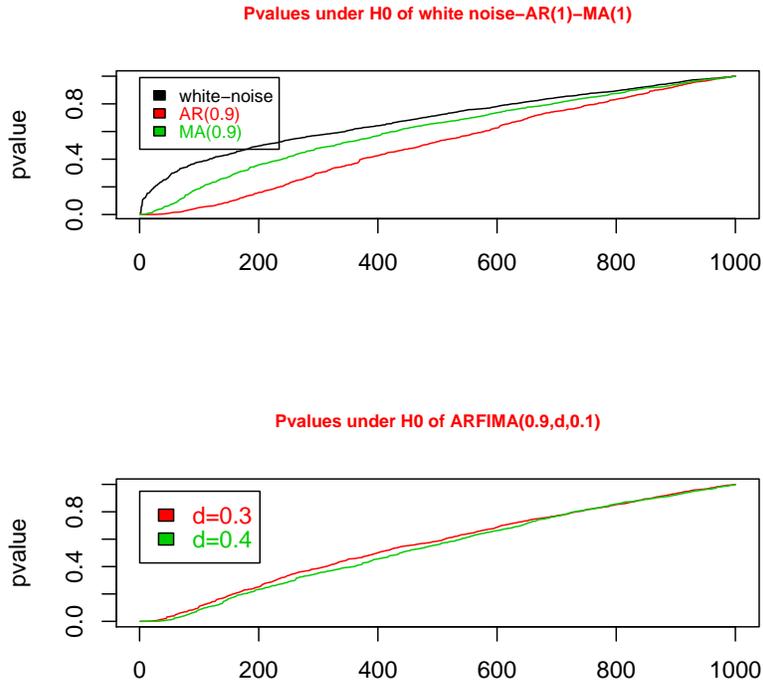}}
\caption{\protect \small Pvalue under $\mathcal{H}_0$ of the distribution $D(J)$ $n=1024$ for white noise and MA(1) processes
  and $n=4096$ for AR(1) and ARFIMA(1,d,1) processes; the coarsest scale is $J=4$ for white noise, MA and AR processes and
  $J=3$ for the ARFIMA process. The finest scale is $J_1=1$.}\label{figresult1}
\end{figure}
\newpage
\subsubsection{Power of $KSM$ and $CVM$.}
We assess the power of test statistic by computing the test statistics in presence of a change in the spectral density. To do so, we consider an observation obtained by concatenation of  $n_1$ observations from a first process and $n_2$ observations from a second process, independent from the first one and having a different spectral density. The length of the resulting observations is $n= n_1+n_2$.
In all cases, we set $n_1=n_2=n/2$, and we present the results for $n_1=512,1024,2048,4096$ and scales $J=4,5$. We consider the following situations: the two processes are white Gaussian noise with two different variances, two AR processes with different values of the autoregressive coefficient, two MA processes with different values of the moving average coefficient and two ARFIMA with same moving average and same autoregressive coefficients but different values of the memory parameter $d$. The scenario considered is a bit artificial but is introduced here to assess the ability of the test to detect abrupt changes in the spectral content. For $1000$ simulations, we report the number of times $\mathcal{H}_1$ was accepted, leading the following results.
\begin{table}\centering
 \begin{tabular}{*{7}{c}}

        \multicolumn{2}{c}{white-noise} & \multicolumn{4}{c}{{ $[\sigma_1^2=1,\, \sigma_2^2=0.7]$}} \\
        \hline\hline
        \multicolumn{2}{c}{$n_1=n_2$}& 512& 1024 & 2048& 4096   \\
        \hline
        $J=4$ &$KSM$ & {0.39}&{ 0.78}& {0.89}&{ 0.95}\\
        $J=4$&$CVM$ & {0.32}&{0.79 }& {0.85}&{0.9 }\\
        \hline
        $J=5$ &$KSM$ & {0.42}&{0.79 }& {0.91}&{0.97 }\\
        $J=5$&$CVM$ & {0.40}&{0.78 }& {0.9}&{0.9 }\\

        \hline\hline

        \end{tabular}
\caption{Power of $\texttt{KSM}-\texttt{CVM}$ on two white noise processes.}
\label{tab6}
\end{table}
\begin{table}\centering
 \begin{tabular}{*{7}{c}}

        \multicolumn{2}{c}{{MA(1)+MA(1)}} & \multicolumn{4}{c}{ $[\theta_1=0.9,\, \theta_2=0.5]$} \\
        \hline\hline
        \multicolumn{2}{c}{$n_1=n_2$}& 512& 1024 & 2048& 4096   \\
        \hline
        $J=4$ &$KSM$ & {0.39}&{0.69}& {0.86}&{0.91}\\
        $J=4$&$CVM$ & {0.31}&{ 0.6}& {0.76}&{ 0.93}\\
        \hline
        $J=5$ &$KSM$ & {0.57}&{ 0.74}& {0.84}&{ 0.94}\\
        $J=5$&$CVM$ & {0.46}&{0.69 }& {0.79}&{0.96 }\\

        \hline\hline

        \end{tabular}
\caption{Power of $\texttt{KSM}-\texttt{CVM}$ on a concatenation of two different $MA$ processes.}
\label{tab7}
\end{table}
\begin{table}\centering
 \begin{tabular}{*{7}{c}}

        \multicolumn{2}{c}{{AR(1)+AR(1)}} & \multicolumn{4}{c}{ $[\phi_1=0.9,\, \phi_2=0.5]$} \\
        \hline\hline
        \multicolumn{2}{c}{$n_1=n_2$}& 512& 1024 & 2048& 4096   \\
        \hline
        $J=4$ &$KSM$ & {0.59}&{ 0.72}& {0.81}&{ 0.87}\\
        $J=4$&$CVM$ & {0.53}&{0.68 }& {0.79}&{0.9 }\\
        \hline
        $J=5$ &$KSM$ & {0.75}&{0.81 }& {0.94}&{0.92 }\\
        $J=5$&$CVM$ & {0.7}&{0.75 }& {0.89}&{0.91}\\

        \hline\hline

        \end{tabular}
\caption{Power of $\texttt{KSM}-\texttt{CVM}$ on a concatenation of two differents $AR$ processes.}
\label{tab8}\end{table}
\begin{table}\centering
 \begin{tabular}{*{7}{c}}

        \multicolumn{1}{c}{{ARFIMA(1,0.3,1)}}&+&{{ARFIMA(1,0.4,1)}} & \multicolumn{3}{c}{ $[\phi=0.9,\,          \theta=0.1]$} \\
        \hline\hline
        \multicolumn{2}{c}{$n_1=n_2$}& 512& 1024 & 2048& 4096   \\
        \hline
        $J=4$ &$KSM$ & {0.86}&{ 0.84}& {0.8}&{ 0.81}\\
        $J=4$&$CVM$ & {0.81}&{ 0.76}& {0.78}&{ 0.76}\\
        \hline
        $J=5$ &$KSM$ & {0.94}&{ 0.94}& {0.9}&{ 0.92}\\
        $J=5$&$CVM$ & {0.93}&{ 0.92}& {0.96}&{ 0.91}\\

        \hline\hline

        \end{tabular}
\caption{Power of $\texttt{KSM}-\texttt{CVM}$ two ARFIMA(1,d,1) with same AR and MA part but two different values of memory parameter $d$.}
\label{tab9}
\end{table}
\begin{figure}[htbp]
\centerline{\includegraphics[width=.8\linewidth]{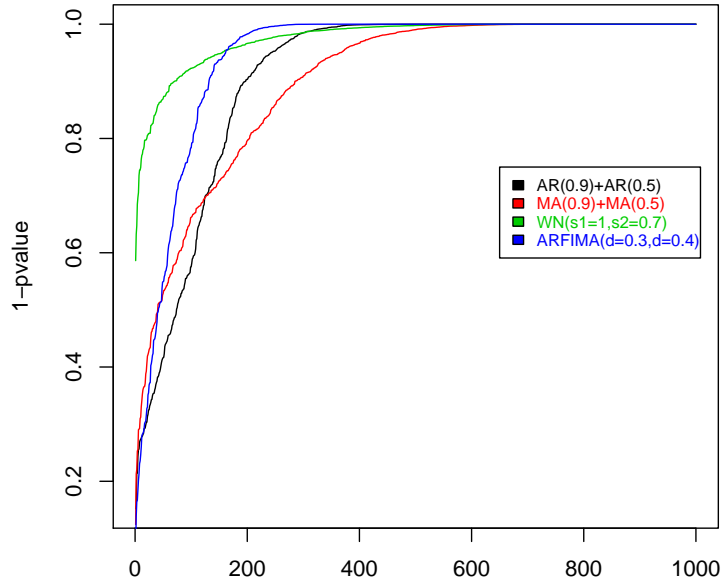}}
\caption{Empirical power of $\mathrm{KSM}(d=4)$ for white noise, AR, MA and ARFIMA processes.}\label{figresult2}
\end{figure}
\\ \indent The power of our two statistics gives us satisfying results for the considered processes, especially if the sample size tends to infinity.
\newpage
\subsubsection{Estimation of the change point in the original process.}
We know that for each scale $j,$ the number $n_j$ of wavelet coefficients is $n_j=2^{-j}(n-\L+1)-\L+1$. If we denote by $k_j$ the change point in the wavelet coefficients at scale $j$ and $k$ the change point in the original signal, then $k=2^j(k_j+\L-1)+\L-1$. In this paragraph, we estimate the change point in the generalized spectral density of a process when it exists and give its 95\% confidence interval. For that, we proceed as before. We consider an observation obtained by concatenation of $n_1$ observations from a first process and $n_2$ observations from a second process, independent from the first one and having a different spectral density. The length of the resulting observations is $n=n_1+n_2$. we estimate the change point in the process and we present the result for $n_1=512,1024,4096,8192$, $n_2=512,2048,8192,$ $J=3$, the statistic $CVM$, two $AR$ processes with different values of the autoregressive coefficient and two $ARFIMA$ with same moving average and same autoregressive coefficients but different values of the memory parameter $d.$ For $10000$ simulations, the bootstrap confidence intervals obtained are set in the tables below. we give also the empirical mean and the median of the estimated change point.
\begin{itemize}
\item $[AR(1),\phi=0.9]$ and $[AR(1),\phi=0.5]$
\begin{table}\centering
 \begin{tabular}{|*{7}{c|}}

        \hline
        $n_1$ &512 & 512&512& 1024&4096&8192\\
        $n_2$ &512&2048 &8192&1024& 4096&8192\\
        \hline
        $MEAN_{CVM}$ &478 &822 & 1853&965&3945&8009\\
        $MEDIAN_{CVM}$ &517 &692 & 1453&1007&4039&8119\\
        $IC_{CVM}$ &[283,661] &[380,1369] & [523,3534]&[637,1350]& [3095,4614]&[7962,8825]\\

        \hline

        \end{tabular}
\caption{Estimation of the change point and confidence interval at 95\% in the generalized spectral density of a process which is obtain by concatenation of two AR(1) processes.}
\label{tab19}
\end{table}

\item $[ARFIMA(1,0.2,1)]$ and $[ARFIMA(1,0.3,1)]$, with $\phi=0.9$ and $\theta=0.2$
\begin{table}\centering
 \begin{tabular}{|*{7}{c|}}

        \hline
        $n_1$ &512 & 512&512& 1024&4096&8192\\
        $n_2$ &512&2048 &8192&1024& 4096&8192\\
        \hline
        $MEAN_{CVM}$ &531 &1162 & 3172&1037&4129&8037\\
        $MEDIAN_{CVM}$ &517 &1115 & 3215&1035&4155&8159\\
        $IC_{CVM}$ &[227,835] &[375,1483]&[817,6300] & [527,1569]& [2985,5830]&[6162,9976]\\

        \hline

        \end{tabular}
\caption{Estimation of the change point and confidence interval at 95\% in the generalized spectral density of a process which is obtain by concatenation of two ARFIMA(1,d,1) processes.}
\label{tab20}
\end{table}
\end{itemize}
We remark that the change point belongs always to the considered confidence interval excepted for $n_1=512,$ $n_2=8192$ where the confidence interval is $[523,3534]$ and the change point $k=512$ doesn't belong it. One can noticed that when the size of the sample increases and
$n_1=n_2$, the interval becomes more accurate. However, as expected, this interval becomes less accurate
when the change appears either at the beginning or at the end of the observations.

\end{document}